\documentclass[sn-mathphys-num]{sn-jnl}

\usepackage{graphicx}%
\usepackage{multirow}%
\usepackage{amsmath,amssymb,amsfonts}%
\usepackage{amsthm}%
\usepackage{mathrsfs}%
\usepackage[title]{appendix}%
\usepackage{xcolor}%
\usepackage{textcomp}%
\usepackage{manyfoot}%
\usepackage{booktabs}%
\usepackage{algorithm}%
\usepackage{algorithmicx}%
\usepackage{algpseudocode}%
\usepackage{listings}%

\theoremstyle{thmstyleone}%
\newtheorem{theorem}{Theorem}
\newtheorem{corollary}[theorem]{Corollary}%
\newtheorem{lemma}[theorem]{Lemma}%

\theoremstyle{thmstyletwo}%
\newtheorem{example}{Example}%

\theoremstyle{thmstylethree}%

\begin{document}

\title[On Galois duality, self-orthogonality, and dual-containment of MP codes]{On Galois duality, self-orthogonality, and dual-containment of matrix product codes}


\author*[1,2]{\fnm{Ramy} \sur{Taki Eldin}}\email{ramy.farouk@eng.asu.edu.eg}

\affil[1]{\orgdiv{Faculty of Engineering}, \orgname{Ain Shams University}, \orgaddress{\city{Cairo}, \country{Egypt}}}

\affil[2]{\orgdiv{Egypt University of Informatics}, \orgname{Knowledge City}, \orgaddress{\city{New Administrative Capital}, \country{Egypt}}}


\abstract{In recent literature, matrix product (MP) codes and their duals have gained significant attention due to their application in the construction of quantum stabilizer codes. In this paper, we begin with providing a formula that characterizes the Galois dual of MP codes. Using this formula, we establish the conditions under which MP codes are self-orthogonal and dual-containing. Although similar results may exist in the literature, the novelty and superiority of our results can be identified in the following points. Previous results that characterize the duals of MP codes only apply to MP codes with an invertible square defining matrix $\mathcal{A}$. However, our characterization applies to MP code with any defining matrix, whether $\mathcal{A}$ is not square or not of full row rank. Previous studies on the conditions for self-orthogonality or dual-containment of MP codes have assumed certain structures for the product $\mathcal{A}\mathcal{A}^T$ or $\mathcal{A}\mathcal{A}^{\dagger}$, such as being diagonal, anti-diagonal, monomial, or partitioned Hermitian orthogonal. However, our conditions do not necessitate such specific structures. Previous studies investigated MP code duality in the context of Euclidean and Hermitian duals; however, we investigate MP code duality in the broader context of Galois dual, with Euclidean and Hermitian duals acting as special cases. Finally, it is worth noting that the proposed conditions for Galois self-orthogonality or dual-containment are both necessary and sufficient. To demonstrate the theoretical results, several numerical examples with best-known parameters MP codes are provided.}




\maketitle
\sloppy
\section{Introduction}
\label{Intro}
In \cite{Blackmore2001}, matrix product (MP) codes over finite fields were introduced as a generalization to Plotkin's well-known $(u|u + v)$-construction and the ternary $(u + v + w|2u + v|u)$-construction. A codeword in an MP code $\mathcal{C}$ is the product of a matrix whose columns are codewords from some constituent codes $\mathcal{C}_1, \mathcal{C}_2, \ldots, \mathcal{C}_M$ and an $M\times N$ defining matrix $\mathcal{A}$. If the constituent codes are linear, then the MP code $\mathcal{C}$ is also linear. Furthermore, if $\mathcal{A}$ is right invertible, then the size of the MP code $\mathcal{C}$ is the product of the sizes of the constituent codes, i.e., $|\mathcal{C}|=|\mathcal{C}_1|\cdots |\mathcal{C}_M|$. In the terminology of \cite{Blackmore2001}, an MP code can be viewed as a generalized concatenated code if $\mathcal{A}$ is non-singular by columns (NSC). Several lower bounds on the minimum distance of MP codes have been proposed in the literature \cite{Blackmore2001,Hernando2009,Hernando2010}. More recently, a novel class of linear codes containing both MP and quasi-twisted codes was introduced in \cite{Eldin2024}, and investigations for the lower bounds on the code minimum distance were conducted as well.

The Euclidean and Hermitian duals of MP codes have been studied extensively in \cite{Blackmore2001,Liu2017,Cao2020,2Cao2020,Jitman,Zhang2024}. These studies have focused on the case where the defining matrix $\mathcal{A}$ is square and invertible. However, the literature lacks a description of these duals for MP codes with non-square or non-full rank defining matrices. In fact, MP codes that are Euclidean or Hermitian dual-containing have received attention in recent years due to their applications in the construction of quantum stabilizer codes \cite{Cao2020,2Cao2020,Liu2017}. Such studies have determined sufficient conditions to ensure the dual-containment for MP codes by assuming a restricted form for the defining matrix $\mathcal{A}$, such as being invertible with a monomial $\mathcal{A}\mathcal{A}^T$ (or $\mathcal{A}\mathcal{A}^{\dagger}$ for Hermitian dual) \cite{Cao2024}. Similarly, self-orthogonal MP codes have also found use in the construction of quantum stabilizer codes \cite{Zhang2024}. It looks natural to explore the conditions under which an MP code is self-orthogonal. Again, current studies in the literature aimed at investigating the self-orthogonality of MP codes under restrictive assumptions on the defining matrix $\mathcal{A}$, such as being invertible with a monomial $\mathcal{A}\mathcal{A}^T$ \cite{Cao2024}, or having full row rank with a diagonal or anti-diagonal $\mathcal{A}\mathcal{A}^{\dagger}$ \cite{JitmanConf,Jitman}. In \cite{Zhang2024}, a sufficient condition for the Hermitian self-orthogonality of MP codes was examined under the assumption that the defining matrix is of full row rank achieving the so-called partitioned Hermitian orthogonal property. In conclusion, the existing literature lacks necessary and sufficient conditions for the self-orthogonality and dual-containment of MP codes with general defining matrices. Further study is required to fully characterize the duality of MP codes with arbitrary defining matrices.

In this paper, we investigate the duality of MP codes in a generalized context. We generalize several aspects compared to the existing literature to make our results more comprehensive. First, we explore the duality of MP codes with defining matrices that are not necessarily invertible or square. Most importantly, we assume any defining matrix, which does not even have to be of full rank. Second, we make no particular assumptions regarding the structure of the product $\mathcal{A}\mathcal{A}^T$ or $\mathcal{A}\mathcal{A}^{\dagger}$, where $\mathcal{A}$ is the defining matrix. Third, the only constraint we impose on the constituent codes is that they must be linear, which is the most important class of error-correcting codes in the literature. Fourth, we investigate the most general context of duality, the Galois dual, from which the Euclidean and Hermitian duals can be derived as special cases. Finally, we aim at making the proposed conditions for the self-orthogonality and dual-containment necessary and sufficient. Overall, we provide a comprehensive investigation for the duality of MP codes without making any restrictive assumptions as those employed in previous studies.

Under the aforementioned setting, we proceed as follows. In Corollary \ref{CoroGaloisDual}, we derive a Galois dual formula for MP codes when the defining matrix has full row rank. This result generalizes \cite{Blackmore2001,Liu2017,Cao2020,2Cao2020,Jitman,Zhang2024}. Whereas Theorem \ref{dualMPCase2} and its subsequent discussion address the situation when the defining matrix is not of full row rank. To the best of our knowledge, this is the first study in the literature to examine MP codes with a non-full rank defining matrix. This result is then employed in Theorem \ref{GaloisSelforthMP_Total} to demonstrate the necessary and sufficient conditions for an MP code to be Galois self-orthogonal. We do not make any particular assumption on the structure of the defining matrix, the product $\mathcal{A}\mathcal{A}^T$, or $\mathcal{A}\mathcal{A}^{\dagger}$. We prove these conditions by considering both the full row rank defining matrix case in Theorem \ref{GaloisSelforthMP} and the non-full row rank case in Subsection \ref{SOwithsmallrank}. Hence we generalize \cite{JitmanConf,Jitman,Zhang2024,Cao2024}. On the other hand, Theorem \ref{GalDualContMP} establishes the necessary and sufficient conditions for an MP code with a full row rank $\mathcal{A}$ to be Galois dual-containing, generalizing the results in \cite{Cao2020,2Cao2020,Liu2017,Cao2024}. However, the case in which $\mathcal{A}$ has non-full row rank is discussed in Theorem \ref{GDualContGen}. To illustrate the application of our theoretical results, we include many numerical examples. We exploit the proposed conditions to run a computer search to present self-orthogonal and dual-containing MP codes with the best-known parameters according to the database \cite{Grassl:codetables}.

The remainder of this paper is organized as follows: Section \ref{Prelim} introduces the fundamental concepts of MP codes and their duals. In Section \ref{EandGdualMP}, we discuss the duals of MP codes within our proposed general settings. Sections \ref{selforthgof MP} and \ref{dualcontnof MP} investigate the conditions for MP codes self-orthogonality and dual-containment, respectively. Finally, Section \ref{concl} concludes our contribution.

\section{Preliminaries}
\label{Prelim}
Let $\mathbb{F}_q$ denote a finite field with $q$ elements. A linear code of length $n$ over $\mathbb{F}_q$ is a subspace of $\mathbb{F}_q^n$ whose dimension is the code dimension, denoted by $k$. Elements of the code are its codewords. A linear code of length $n$ and dimension $k$ is said to have the parameters $[n,k,d]$, where $d$ is the code minimum distance. We indicate that $\mathcal{A}$ is a matrix of size $M\times N$ over $\mathbb{F}_q$ by $\mathcal{A}\in\mathbb{F}_q^{M\times N}$. Consider a collection of $M$ linear codes $\mathcal{C}_1, \mathcal{C}_2, \ldots,\mathcal{C}_M$ over $\mathbb{F}_q$ of length $n$. An MP code $\mathcal{C}$ with defining matrix $\mathcal{A}\in\mathbb{F}_q^{M\times N}$ and constituent codes $\mathcal{C}_1, \mathcal{C}_2, \ldots,\mathcal{C}_M$ is a linear code of length $nN$ whose codewords are specified as follows. For each choice of constituent codewords $\mathbf{c}_1\in\mathcal{C}_1, \mathbf{c}_2\in\mathcal{C}_2, \ldots,\mathbf{c}_M\in\mathcal{C}_M$, there exists a corresponding codeword $\mathbf{c}\in\mathcal{C}$ obtained by vectorizing the matrix resulting from the product $$\left[\mathbf{c}_1 \ \mathbf{c}_2 \ \cdots \ \mathbf{c}_M\right]\mathcal{A},$$ where $\left[\mathbf{c}_1 \ \mathbf{c}_2 \ \cdots \ \mathbf{c}_M\right]$ is the matrix with the constituent codewords as columns. The MP code with defining matrix $\mathcal{A}$ and constituent codes $\mathcal{C}_1, \mathcal{C}_2, \ldots,\mathcal{C}_M$ is denoted by $$\mathcal{C}=\left[\mathcal{C}_1 \ \mathcal{C}_2 \ \cdots \ \mathcal{C}_M\right]\mathcal{A}.$$
In the literature, an MP codeword can be expressed in the matrix form $\left[\mathbf{c}_1 \ \mathbf{c}_2 \ \cdots \ \mathbf{c}_M\right]\mathcal{A}$ or the vectorization of this matrix form.

An MP code, as a linear code, has a generator matrix. Let $a_{i,j}$ be the $(i,j)^{\text{th}}$ entry of $\mathcal{A}$, written $\mathcal{A}=[a_{i,j}]$, and let $\mathbf{G}_1, \mathbf{G}_2, \ldots, \mathbf{G}_M$ be generator matrices of the constituent codes. According to \cite{Blackmore2001}, a generator matrix for the MP code $\mathcal{C}=\left[\mathcal{C}_1 \ \mathcal{C}_2 \ \cdots \ \mathcal{C}_M\right]\mathcal{A}$ is given by 
\begin{equation*}
\mathbf{G}=\begin{bmatrix}
 \mathbf{G}_1 a_{1,1} & \mathbf{G}_1 a_{1,2} & \cdots & \mathbf{G}_1 a_{1,N}\\
 \mathbf{G}_2 a_{2,1} & \mathbf{G}_2 a_{2,2} & \cdots & \mathbf{G}_2 a_{2,N}\\
 \vdots & \vdots & \ddots & \vdots\\
 \mathbf{G}_M a_{M,1} & \mathbf{G}_M a_{M,2} & \cdots & \mathbf{G}_M a_{M,N}
\end{bmatrix}.
\end{equation*}
From \cite{Eldin2024}, this may be written compactly as 
\begin{equation}
\label{MP_Gen}
\mathbf{G}=\text{diag}\left[\mathbf{G}_1 \ \cdots \ \mathbf{G}_M\right] ( \mathcal{A}\otimes \mathcal{I}_n),
\end{equation}
where $\otimes$ is the Kronecker product, $\mathcal{I}_n$ is the identity matrix of size $n$, and $\text{diag}\left[\mathbf{G}_1 \ \cdots \ \mathbf{G}_M\right]$ denotes the block diagonal matrix with block entries $\mathbf{G}_1, \mathbf{G}_2, \ldots, \mathbf{G}_M$. In accordance with \eqref{MP_Gen}, if $\mathcal{A}$ is the identity matrix, the MP code is the direct sum of the constituent codes. The following notational conventions will be used consistently throughout this paper. The zero code is denoted by $\mathcal{O}$ and it is the code of length $n$ and zero dimension, whereas the whole space code is denoted by $\mathcal{F}$ and it is the code of length $n$ and dimension $n$. The zero matrix is denoted by $\mathbf{0}$ and it is the matrix with entirely zero entries. For any matrix $\mathcal{A}$, the transpose of $\mathcal{A}$ is denoted by $\mathcal{A}^T$. We frequently write $\mathcal{A}\{i_1,i_2,\ldots,i_\tau\}$ to denote the submatrix of $\mathcal{A}$ formed by the rows $i_1, i_2, \ldots, i_\tau$. A matrix $\mathcal{A}\in\mathbb{F}_q^{M\times N}$ is called non-singular by columns (NSC) if for every $1\le i\le M$, each $i \times i$ submatrix of $\mathcal{A}\{1,2,\ldots,i\}$ is non-singular.

There have been several attempts to determine lower bounds on the minimum distance $d$ of an MP code. The most well-known bound is given in \cite[Theorem 3.7]{Blackmore2001} for NSC defining matrix. Specifically, if $\mathcal{C}=\left[\mathcal{C}_1 \ \mathcal{C}_2 \ \cdots \ \mathcal{C}_M\right]\mathcal{A}$, where $\mathcal{A}\in\mathbb{F}_q^{M\times N}$ is NSC and $\mathcal{C}_i$ is $[n,k_i,d_i]$ for $1\le i\le M$, then 
$$d\ge \mathrm{min}\left\{N d_1, (N-1)d_2, \ldots, (N-M+1)d_M\right\}.$$ 
Other lower bounds on $d$ exist that do not require an NSC defining matrix. For instance, in \cite[Lemma 3]{Cao2020}, if $\mathcal{A}$ is of full row rank, then 
$$d\ge \mathrm{min}\left\{d_1 D_1(\mathcal{A}), d_2 D_2(\mathcal{A}), \ldots, d_M D_M(\mathcal{A})\right\},$$
where, for $1\le i\le M$, $ D_i(\mathcal{A})$ is the minimum distance of the linear code over $\mathbb{F}_q$ of length $N$ generated by $\mathcal{A}\{1,2,\ldots,i\}$.

Assume $q=p^e$ for a prime $p$ and a positive integer $e$. Denote the Frobenius automorphism of $\mathbb{F}_q$ by $\sigma$, that is, $\sigma:\alpha \mapsto \alpha^p$ for $\alpha\in\mathbb{F}_q$. We naturally extend $\sigma$ to act on vectors of $\mathbb{F}_q^n$ and matrices over $\mathbb{F}_q$ by applying it componentwise. For instance, if $\mathcal{A}=[a_{i,j}]\in\mathbb{F}_q^{M\times N}$, then $\sigma\left(\mathcal{A}\right)=\left[\sigma\left(a_{i,j}\right)\right]\in\mathbb{F}_q^{M\times N}$. Similarly, if $\mathcal{C}$ is a linear code over $\mathbb{F}_q$, we define $$\sigma\left(\mathcal{C}\right)=\left\{\sigma\left(\mathbf{c}\right) \ \forall \ \mathbf{c}\in\mathcal{C} \right\}.$$ 
The linearity of $\sigma$ implies that $\sigma\left(\mathcal{C}\right)$ is linear and generated by $\sigma\left(\mathbf{G}\right)$ for any generator matrix $\mathbf{G}$ of $\mathcal{C}$. The Euclidean and Hermitian inner products on $\mathbb{F}_q^n$ can be generalized to the Galois inner product. For some $0\le \ell < e$, the $\ell$-Galois inner product of $\mathbf{a}=\left(a_1,\ldots,a_n\right)$ and $\mathbf{b}=\left(b_1,\ldots,b_n\right)\in\mathbb{F}_q^n$ is defined as 
$$\langle \mathbf{a},\mathbf{b}\rangle_\ell=\sum_{i=1}^n a_i \sigma^\ell\left(b_i\right)=\sum_{i=1}^n a_i  b_i^{p^{\ell}}.$$
Obviously, the Euclidean inner product is the $0$-Galois inner product. However, for even $e$, the Hermitian inner product is the $\frac{e}{2}$-Galois inner product. Thus, the Galois dual can be defined for linear codes just as the Euclidean and Hermitian duals were defined. Specifically, let $\mathcal{C}$ be a linear code over $\mathbb{F}_q$ and let $\ell$ be a non-negative integer less than $e$. The $\ell$-Galois dual of $\mathcal{C}$, denoted by $\mathcal{C}^{\perp_\ell}$, is defined as 
$$\mathcal{C}^{\perp_\ell}=\left\{\mathbf{a}\in\mathbb{F}_q^n \text{ such that } \langle \mathbf{c},\mathbf{a}\rangle_\ell=0 \text{ for every } \mathbf{c}\in\mathcal{C}\right\}.$$
Again, the $\ell$-Galois dual $\mathcal{C}^{\perp_\ell}$ generalizes both Euclidean and Hermitian duals. Specifically, the Euclidean dual $\mathcal{C}^\perp$ is the $0$-Galois dual, while the Hermitian dual is the $\frac{e}{2}$-Galois dual when $e$ is even. Several interesting properties of the Galois dual of linear codes are discussed in \cite{TakiEldin2023}. The properties that are most important to this study are:
\begin{enumerate}
\item $\langle \mathbf{c},\mathbf{a}\rangle_\ell=\langle \mathbf{c},\sigma^\ell\left(\mathbf{a}\right)\rangle_0$.
\item $\mathcal{C}^{\perp_\ell}=\sigma^{e-\ell}\left(\mathcal{C}^\perp\right)=\left(\sigma^{e-\ell}\mathcal{C}\right)^\perp$.
\item $\left(\mathcal{C}^{\perp_\ell}\right)^{\perp_{e-\ell}}=\left(\mathcal{C}^{\perp_{e-\ell}}\right)^{\perp_{\ell}}=\mathcal{C}$.
\end{enumerate}

\section{Galois dual of MP cods}
\label{EandGdualMP}
The main aim of this section is to obtain a formula for the Galois dual of MP codes. To do this, we have to consider two separate cases: when the code defining matrix has full row rank and when it does not. Hereinafter, we use $\mathcal{A}\in\mathbb{F}_q^{M\times N}$ as the defining matrix of the code.

\subsection{$\mathcal{A}$ of rank $M$}
\label{EandGdualMP_S1}
Throughout this subsection we assume that $\mathcal{C}$ is an MP code with a defining matrix $\mathcal{A}\in\mathbb{F}_q^{M\times N}$ of rank $M$, that is, $\mathcal{A}$ has full row rank. However, the other case will be considered in Section \ref{MlargerN}. We begin by characterizing the Euclidean dual $\mathcal{C}^\perp$, then proceed to characterize the $\ell$-Galois dual $\mathcal{C}^{\perp_\ell}$. Recall that $\mathcal{F}$ denotes the whole space code of length $n$ over $\mathbb{F}_q$, whereas the zero code is denoted by $\mathcal{O}$.

\begin{theorem}
\label{dualMP}
Assume that $\mathcal{A}\in \mathbb{F}_q^{M\times N}$ is of rank $M$ and let $\mathcal{C}_i$, $i=1,2,\ldots, M$, be linear codes over $\mathbb{F}_q$ of length $n$. The Euclidean dual of $\mathcal{C}=\left[\mathcal{C}_1 \ \mathcal{C}_2 \ \cdots \ \mathcal{C}_M\right]\mathcal{A}$ is the MP code
\begin{equation*}
\mathcal{C}^\perp= [\mathcal{C}^\perp_1 \ \mathcal{C}^\perp_2 \ \cdots \ \mathcal{C}^\perp_M \ \overbrace{\mathcal{F} \ \cdots \ \mathcal{F}}^{\text{$N-M$ times}} ]\left(\mathcal{B}^{-1}\right)^T,
\end{equation*}
where $\mathcal{B}$ is any $N\times N$ invertible matrix over $\mathbb{F}_q$ such that $\mathcal{B}\{1,\ldots,M\}=\mathcal{A}$.
\end{theorem}
\begin{proof}
One may rewrite $\mathcal{C}$ as an MP code with defining matrix $\mathcal{B}$, namely,
\begin{equation*}
\mathcal{C}= \left[\mathcal{C}_1 \ \mathcal{C}_2 \ \cdots \ \mathcal{C}_M\right]\mathcal{A}= \left[\mathcal{C}_1 \ \mathcal{C}_2 \ \cdots \ \mathcal{C}_M \ \mathcal{O} \ \cdots \ \mathcal{O}\right]\mathcal{B}.
\end{equation*}
Since $\mathcal{B}$ is square invertible and $\mathcal{O}^\perp=\mathcal{F}$, the result follows from \cite[Proposition 6.2]{Blackmore2001}.
\end{proof}

\begin{example}
Consider the MP code $\mathcal{C}=\left[\mathcal{C}_1 \ \mathcal{C}_2 \ \mathcal{C}_3\right]\mathcal{A}$ over $\mathbb{F}_5$, where $\mathcal{C}_1$, $\mathcal{C}_2$, and $\mathcal{C}_3$ are the linear codes over $\mathbb{F}_5$ of length $5$ generated by 
\begin{equation*}
\mathbf{G}_1=\begin{bmatrix}
 1 & 0 & 2 & 4 & 0\\
 0 & 1 & 1 & 3 & 2
\end{bmatrix}\quad,
\quad
\mathbf{G}_2=\begin{bmatrix}
 1 & 0 & 4 & 2 & 4\\
 0 & 1 & 2 & 4 & 4
\end{bmatrix}\quad\text{, and}
\quad
\mathbf{G}_3=\begin{bmatrix}
 1 & 4 & 4 & 4 & 1
\end{bmatrix},
\end{equation*}
respectively, and $\mathcal{A}$ is the NSC matrix
\begin{equation*}
\mathcal{A}=\begin{bmatrix}
 4 & 1 & 1 & 3\\
 3 & 3 & 1 & 2\\
 1 & 4 & 3 & 4
\end{bmatrix}.
\end{equation*}
One possible choice for $\mathcal{B}\in\mathbb{F}_5^{4\times 4}$ that preserves the first three rows of $\mathcal{A}$ is
\begin{equation*}
\mathcal{B}=\begin{bmatrix}
 4 & 1 & 1 & 3\\
 3 & 3 & 1 & 2\\
 1 & 4 & 3 & 4\\
 1 & 0 & 0 & 0
\end{bmatrix}.
\end{equation*}
By Theorem \ref{dualMP}, 
$$\mathcal{C}^\perp= [\mathcal{C}^\perp_1 \ \mathcal{C}^\perp_2 \  \mathcal{C}^\perp_3 \ \mathcal{F}]\left(\mathcal{B}^{-1}\right)^T.$$
In fact, $\mathcal{C}$ has the parameters $[20,5,12]$ over $\mathbb{F}_5$, while $\mathcal{C}^\perp$ has the best-known parameters $[20,15,4]$ according to the database \cite{Grassl:codetables}.
\end{example}

Theorem \ref{dualMP} generalizes the previous results \cite{Blackmore2001, Cao2020, 2Cao2020} by not requiring $\mathcal{A}$ to be square invertible. Despite its simplicity, Theorem \ref{dualMP} is the first in the literature to characterize the dual of an MP code with any full row rank defining matrix. The Galois dual of MP codes may now be easily concluded from the Euclidean dual.

\begin{corollary}
\label{CoroGaloisDual}
Assume that $\mathcal{A}\in \mathbb{F}_q^{M\times N}$ is of rank $M$, where $q=p^e$, and let $\mathcal{B}$ be any $N\times N$ invertible matrix over $\mathbb{F}_q$ such that $\mathcal{B}\{1,\ldots,M\}=\mathcal{A}$. Let $\mathcal{C}_i$, $i=1,2,\ldots,M$, be linear codes over $\mathbb{F}_q$ of length $n$. For any $0\le \ell < e$, the $\ell$-Galois dual of $\mathcal{C}=\left[\mathcal{C}_1 \ \mathcal{C}_2 \ \cdots \ \mathcal{C}_M\right]\mathcal{A}$ is the MP code
\begin{equation*}
\mathcal{C}^{\perp_\ell}= [\mathcal{C}^{\perp_\ell}_1 \ \mathcal{C}^{\perp_\ell}_2 \ \cdots \ \mathcal{C}^{\perp_\ell}_M \ \overbrace{\mathcal{F} \ \cdots \ \mathcal{F}}^{\text{$N-M$ times}} ]\left(\sigma^{e-\ell}\left(\mathcal{B}\right)^{-1}\right)^T.
\end{equation*}
\end{corollary}
\begin{proof}
Using the properties $\mathcal{C}^{\perp_\ell}=\sigma^{e-\ell}\left(\mathcal{C}^\perp\right)$ and $\mathcal{C}_i^{\perp_\ell}=\sigma^{e-\ell}\left(\mathcal{C}_i^\perp\right)$, we conclude from Theorem \ref{dualMP} that
\begin{equation*}
\begin{split}
\mathcal{C}^{\perp_\ell}&=\sigma^{e-\ell}\left([\mathcal{C}^\perp_1 \ \mathcal{C}^\perp_2 \ \cdots \ \mathcal{C}^\perp_M \ \mathcal{F} \ \cdots \ \mathcal{F} ]\left(\mathcal{B}^{-1}\right)^T \right)\\
&=\left[\sigma^{e-\ell}\left(\mathcal{C}^\perp_1\right) \ \sigma^{e-\ell}\left(\mathcal{C}^\perp_2\right) \ \cdots \ \sigma^{e-\ell}\left(\mathcal{C}^\perp_M\right) \ \mathcal{F} \ \cdots \ \mathcal{F} \right]\left(\sigma^{e-\ell}\left(\mathcal{B}\right)^{-1}\right)^T \\
&=\left[\mathcal{C}^{\perp_\ell}_1 \ \mathcal{C}^{\perp_\ell}_2 \ \cdots \ \mathcal{C}^{\perp_\ell}_M \ \mathcal{F} \ \cdots \ \mathcal{F} \right]\left(\sigma^{e-\ell}\left(\mathcal{B}\right)^{-1}\right)^T .
\end{split}
\end{equation*}
\end{proof}

\begin{example}
Let $\theta\in\mathbb{F}_8$ be such that $\theta^3+\theta+1=0$ and consider the MP code $\mathcal{C}=\left[\mathcal{C}_1 \ \mathcal{C}_2 \right]\mathcal{A}$ over $\mathbb{F}_8$ with parameters $[50,9,20]$, where $\mathcal{C}_1$ and $\mathcal{C}_2$ are the linear codes over $\mathbb{F}_8$ of length $10$ generated by 
\begin{equation*}
\mathbf{G}_1=\begin{bmatrix}
   1 & 0 & 0 & 0 & 0 & \theta & \theta & \theta^4 & \theta^2 & 0\\
   0 & 1 & 0 & 0 & 0 & 1 & \theta & \theta^5 & 0 & \theta^6\\
   0 & 0 & 1 & 0 & 0 & 1 & \theta & 0 & 1 & 1\\
   0 & 0 & 0 & 1 & 0 & 1 & \theta^6 & 0 & \theta & \theta^6\\
   0 & 0 & 0 & 0 & 1 & \theta^6 & \theta^5 & 1 & \theta^6 & \theta^4
\end{bmatrix}\quad \text{and}
\quad
\mathbf{G}_2=\begin{bmatrix}
   1 & 0 & 0 & 0 & \theta^2 & 1 & \theta & \theta^4 & 1 & \theta^5\\
   0 & 1 & 0 & 0 & \theta^2 & \theta^2 & 1 & \theta & \theta^6 & 1\\
   0 & 0 & 1 & 0 & 1 & \theta & 0 & 1 & \theta & \theta^2\\
   0 & 0 & 0 & 1 & \theta^4 & \theta^3 & 1 & \theta^2 & 1 & \theta^5
\end{bmatrix},
\end{equation*}
respectively, and 
\begin{equation*}
\mathcal{A}=\begin{bmatrix}
   0 & 1 & \theta^4 & 1 & \theta\\
 \theta^6 & \theta^2 & \theta^3 & \theta^5 & 0
\end{bmatrix}.
\end{equation*}
One possible choice for an invertible $\mathcal{B}\in\mathbb{F}_8^{5\times 5}$ that preserves the first two rows of $\mathcal{A}$ is
\begin{equation*}
\mathcal{B}=\begin{bmatrix}
   0 & 1 & \theta^4 & 1 & \theta\\
 \theta^6 & \theta^2 & \theta^3 & \theta^5 & 0\\
 1 & 0 & 0 & 0 & 0\\
 0 & 1 & 0 & 0 & 0\\
 0 & 0 & 1 & 0 & 0
\end{bmatrix}.
\end{equation*}
The $2$-Galois dual of $\mathcal{C}$ is obtained by Corollary \ref{CoroGaloisDual} as
\begin{equation*}
\mathcal{C}^{\perp_2}=[\mathcal{C}^{\perp_2}_1 \ \mathcal{C}^{\perp_2}_2 \ \mathcal{F} \ \mathcal{F} \ \mathcal{F}]\left(\sigma\left(\mathcal{B}\right)^{-1}\right)^T=[\mathcal{C}^{\perp_2}_1 \ \mathcal{C}^{\perp_2}_2 \ \mathcal{F} \ \mathcal{F} \ \mathcal{F}] \begin{bmatrix}
   0 &0&  0&  0&\theta^5\\
   0&  0&  0&\theta^4&\theta^2\\
   1&  0&  0&\theta^2&  1\\
   0&  1&  0&\theta&\theta\\
   0&  0&  1&\theta^3&\theta^5
\end{bmatrix}.
\end{equation*}
In fact, $\mathcal{C}^{\perp_2}$ has the parameters $[50,41,3]$ over $\mathbb{F}_8$. In addition, $\mathcal{C}^{\perp_2}_1=\sigma\left(\mathcal{C}^\perp_1 \right)$ and $\mathcal{C}^{\perp_2}_2=\sigma\left(\mathcal{C}^\perp_2 \right)$ are generated by
\begin{equation*}
\sigma(\mathbf{H}_1)=\begin{bmatrix}
1&0&0&0&0&\theta&1&\theta&\theta^4&0\\
0&1&0&0&0&\theta^5&0&1&\theta^6&\theta\\
0&0&1&0&0&\theta^6&\theta^6&\theta^5&\theta^2&\theta\\
0&0&0&1&0&\theta^2&\theta^2&1&\theta^4&\theta^2\\
0&0&0&0&1&\theta^4&\theta&\theta^4&\theta^3&\theta^4
\end{bmatrix} \quad \text{ and }\quad
\sigma(\mathbf{H}_2)=\begin{bmatrix}
1&0&0&0&0&0&\theta^4&\theta^6&\theta^5&\theta^5\\
0&1&0&0&0&0&\theta&\theta^5&0&\theta\\
0&0&1&0&0&0&\theta^4&\theta&\theta^3&\theta^5\\
0&0&0&1&0&0&\theta^4&\theta^6&1&\theta^3\\
0&0&0&0&1&0&\theta^2&\theta^2&\theta^6&\theta\\
0&0&0&0&0&1&\theta^5&\theta^6&\theta^6&\theta^6
\end{bmatrix},
\end{equation*}
where $\mathbf{H}_1$ and $\mathbf{H}_2$ are parity-check matrices of $\mathcal{C}_1$ and $\mathcal{C}_2$, respectively.
\end{example}

\subsection{$\mathcal{A}$ of rank $<M$}
\label{MlargerN}
In this subsection, we investigate the Galois dual of an MP code with a defining matrix $\mathcal{A}\in\mathbb{F}_q^{M\times N}$ of rank less than $M$, which means $\mathcal{A}$ has non-full row rank. As demonstrated in Corollary \ref{CoroGaloisDual}, the Galois dual of an MP code with full row rank $\mathcal{A}$ is MP. On the contrary, we prove here that such a dual is not necessary MP when $\mathcal{A}$ has rank less than $M$. This fact was mentioned in \cite{Jitman} on the Hermitian dual of MP codes, but no explicit evidence was provided. The following result is a more in-depth investigation that provides a proof for this fact. Recall that, for $\tau$ distinct integers $1\le i_1, i_2, \ldots, i_\tau \le M$, the $\tau\times N$ submatrix of $\mathcal{A}$ formed of the rows $i_1, i_2, \ldots, i_\tau$ is denoted by $\mathcal{A}\{i_1, i_2, \ldots, i_\tau\}$.

\begin{theorem}
\label{dualMPCase2}
Let $\mathcal{C}=\left[\mathcal{C}_1 \ \mathcal{C}_2 \ \cdots \ \mathcal{C}_M\right]\mathcal{A}$, where $\mathcal{A}\in \mathbb{F}_q^{M\times N}$ and $\mathcal{C}_i$, $i=1, 2,\ldots, M$, are linear codes of length $n$ over $\mathbb{F}_q$. Assume that the set $\{1, 2, \ldots, M\}$ is partitioned into two disjoint subsets $\{i_1, i_2, \ldots, i_{\tau}\}$ and $\{i_{\tau+1}, i_{\tau+2}, \ldots, i_M\}$. Define two MP codes of the same length as $\mathcal{C}$ as follows
\begin{equation*}
\begin{split}
\mathcal{C}^{(1)}&=\left[\mathcal{C}_{i_1} \ \mathcal{C}_{i_2} \ \cdots \ \mathcal{C}_{i_{\tau}}\right]\mathcal{A}\{i_1, i_2, \ldots, i_{\tau}\},\\
\mathcal{C}^{(2)}&=\left[\mathcal{C}_{i_{\tau+1}} \ \mathcal{C}_{i_{\tau+2}} \ \cdots \ \mathcal{C}_{i_M}\right]\mathcal{A}\{i_{\tau+1}, i_{\tau+2}, \ldots, i_M\}.
\end{split}
\end{equation*}
Then $\mathcal{C}=\mathcal{C}^{(1)} + \mathcal{C}^{(2)}$, $\mathcal{C}^{\perp}=\mathcal{C}^{(1)\perp} \cap \mathcal{C}^{(2) \perp}$, and  $\mathcal{C}^{\perp_\ell}=\mathcal{C}^{(1) \perp_\ell} \cap \mathcal{C}^{(2) \perp_\ell}$.
\end{theorem}
\begin{proof}
By permuting the constituent codes $\mathcal{C}_1, \mathcal{C}_2, \ldots, \mathcal{C}_M$ simultaneously with the rows of $\mathcal{A}$, we have 
\begin{equation*}
\begin{split}
\mathcal{C}&=\left[\mathcal{C}_1 \ \mathcal{C}_2 \ \cdots \ \mathcal{C}_M\right]\mathcal{A}\\
&=\left[\mathcal{C}_{i_1} \ \mathcal{C}_{i_2} \ \cdots \ \mathcal{C}_{i_{\tau}} \ \mathcal{C}_{i_{\tau+1}} \ \mathcal{C}_{i_{\tau+2}} \ \cdots \ \mathcal{C}_{i_M}\right] \left[\begin{matrix} \mathcal{A}\{i_1, i_2, \ldots, i_{\tau}\} \\ \mathcal{A}\{i_{\tau+1}, i_{\tau+2}, \ldots, i_M\}\end{matrix}\right]\\
&=\left[\mathcal{C}_{i_1} \ \mathcal{C}_{i_2} \ \cdots \ \mathcal{C}_{i_{\tau}}\right]\mathcal{A}\{i_1, i_2, \ldots, i_{\tau}\}+\left[\mathcal{C}_{i_{\tau+1}} \ \mathcal{C}_{i_{\tau+2}} \ \cdots \ \mathcal{C}_{i_M}\right]\mathcal{A}\{i_{\tau+1}, i_{\tau+2}, \ldots, i_M\}\\
&= \mathcal{C}^{(1)} + \mathcal{C}^{(2)}.
\end{split}
\end{equation*}
Consequently, elementary properties of linear codes imply that $\mathcal{C}^\perp= \left(\mathcal{C}^{(1)} + \mathcal{C}^{(2)}\right)^\perp= \mathcal{C}^{(1) \perp} \cap \mathcal{C}^{(2) \perp}$, and Galois dual properties show that  
\begin{equation*}
\begin{split}
\mathcal{C}^{\perp_\ell}=\sigma^{e-\ell}\left(\mathcal{C}^\perp\right)&=\sigma^{e-\ell}\left(\mathcal{C}^{(1) \perp} \cap \mathcal{C}^{(2) \perp}\right)\\&=\sigma^{e-\ell}\left(\mathcal{C}^{(1) \perp}\right) \cap \sigma^{e-\ell}\left(\mathcal{C}^{(2) \perp}\right)= \mathcal{C}^{(1) \perp_\ell} \cap \mathcal{C}^{(2) \perp_\ell}.
\end{split}
\end{equation*}
\end{proof}

By repeatedly applying Theorem \ref{dualMPCase2}, we conclude that partitioning the rows of the defining matrix $\mathcal{A}$ of an MP code allows the code to be expressed as a sum of MP codes, each with a defining matrix that is a submatrix of $\mathcal{A}$. Subsection \ref{EandGdualMP_S1} gives a complete description of the dual of any MP code with a full row rank defining matrix. Therefore, when applying Theorem \ref{dualMPCase2}, it is advantageous to partition the set of rows of $\mathcal{A}$ into subsets, each of which is independent. This allows the MP code to be written as a sum of MP codes, each of which has a dual as described in Subsection \ref{EandGdualMP_S1}. In our notations, an MP code $\mathcal{C}$ with a defining matrix $\mathcal{A}\in \mathbb{F}_q^{M\times N}$ of rank $\le M$ can be written as a sum of MP codes $\mathcal{C}^{(t)}$, where each $\mathcal{C}^{(t)}$ has a defining matrix $\mathcal{A}_t \in \mathbb{F}_q^{M_t \times N}$ of rank $M_t$. To do this, we partition $\mathcal{A}$ into sub-matrices $\mathcal{A}_t$, each of which consists of $M_t$ linearly independent rows. Now $\mathcal{C}=\sum_t \mathcal{C}^{(t)}$. Since each $\mathcal{C}^{(t)}$ has a full row rank defining matrix, its dual is completely determined by Subsection \ref{EandGdualMP_S1}. It follows from Theorem \ref{dualMPCase2} that $\mathcal{C}^\perp=\cap_t \mathcal{C}^{(t)\perp}$ and $\mathcal{C}^{\perp_\ell}=\cap_t \mathcal{C}^{(t)\perp_\ell}$. For instance, let $\mathcal{C}=\left[\mathcal{C}_1 \ \mathcal{C}_2 \ \mathcal{C}_3 \ \mathcal{C}_4 \ \mathcal{C}_5 \right]\mathcal{A}$, where
\begin{equation*}
\mathcal{A}= \begin{bmatrix}
   1 &1\\
   1& 1\\
   1&  0\\
   0&  1\\
   1& 0
\end{bmatrix}\in\mathbb{F}_2^{5 \times 2}.
\end{equation*}
If we set 
\begin{equation*}
\mathcal{A}_1=\mathcal{A}\{1,3\}= \begin{bmatrix}
   1 &1\\
   1&  0
\end{bmatrix}\quad , \quad 
\mathcal{A}_2=\mathcal{A}\{2,4\}= \begin{bmatrix}
   1& 1\\
   0&  1
   \end{bmatrix}\quad , \text{ and}\quad
\mathcal{A}_3=\mathcal{A}\{5\}= \begin{bmatrix}
   1& 0
\end{bmatrix},
\end{equation*}
then $\mathcal{C}=\mathcal{C}^{(1)}+\mathcal{C}^{(2)}+\mathcal{C}^{(3)}$, where $\mathcal{C}^{(1)}=\left[\mathcal{C}_1 \ \mathcal{C}_3\right]\mathcal{A}_1$, $\mathcal{C}^{(2)}=\left[\mathcal{C}_2 \ \mathcal{C}_4\right]\mathcal{A}_2$, and $\mathcal{C}^{(3)}=\left[\mathcal{C}_5\right]\mathcal{A}_3$. Then $\mathcal{C}^\perp=\mathcal{C}^{(1)\perp}\cap\mathcal{C}^{(2)\perp}\cap\mathcal{C}^{(3)\perp}$. We have seen in Theorem \ref{dualMP} that 
\begin{equation*}
\mathcal{C}^{(1)\perp}=\left[\mathcal{C}_1^\perp \ \mathcal{C}_3^\perp\right]\left(\mathcal{A}_1^{-1}\right)^T  \quad,\quad \mathcal{C}^{(2)\perp}=\left[\mathcal{C}_2^\perp \ \mathcal{C}_4^\perp\right]\left(\mathcal{A}_2^{-1}\right)^T \quad \text{, and}\quad \mathcal{C}^{(3)\perp}=\left[\mathcal{C}_5^\perp \ \mathcal{F}\right]\mathcal{I}_2.
\end{equation*}

It is worth noting that Theorem \ref{dualMPCase2} explains why the dual of an MP code with a non-full row rank defining matrix is not generally MP. This is due to the fact that the intersection of MP codes is not necessarily MP.

\begin{example}
Consider the binary MP code
\begin{equation*}
\mathcal{C}=\left[\mathcal{C}_1 \ \mathcal{C}_2 \ \mathcal{C}_3 \ \mathcal{C}_4 \right] \begin{bmatrix}
   0& 1\\
   0& 1\\
   1& 0\\
   1&  1
   \end{bmatrix},
\end{equation*}
where $\mathcal{C}_1, \mathcal{C}_2, \mathcal{C}_3, \mathcal{C}_4$ are the linear binary codes of length $5$ generated by
\begin{align*}
\mathbf{G}_1&=\begin{bmatrix}
 0 & 1 & 0 & 1 & 1\\
 0 & 0 & 1 & 1 & 0
\end{bmatrix} \quad , \quad
\mathbf{G}_2=\begin{bmatrix}
 0 & 1 & 1 & 0 & 1
\end{bmatrix},\\
\mathbf{G}_3&=\begin{bmatrix}
 1 & 1 & 0 & 1 & 0
\end{bmatrix} \quad , \quad
\mathbf{G}_4=\begin{bmatrix}
 1 & 0 & 0 & 0 & 1\\
 0 & 1 & 0 & 0 & 1\\
 0 & 0 & 1 & 0 & 1\\
 0 & 0 & 0 & 1 & 1
\end{bmatrix}.
\end{align*}
According to \cite{Grassl:codetables}, $\mathcal{C}$ has the best-known parameters $[10,7,2]$. By Theorem \ref{dualMPCase2}, 
\begin{equation*}
\mathcal{C}=\left[\mathcal{C}_1 \ \mathcal{C}_3 \right] \begin{bmatrix}
   0& 1\\
   1& 0
   \end{bmatrix}+\left[\mathcal{C}_2 \ \mathcal{C}_4 \right] \begin{bmatrix}
   0& 1\\
   1&  1
   \end{bmatrix}.
\end{equation*}
Now we see from Theorems \ref{dualMP} and \ref{dualMPCase2} that
\begin{equation*}
\mathcal{C}^\perp=\left[\mathcal{C}_1^\perp \ \mathcal{C}_3^\perp \right] \begin{bmatrix}
   0& 1\\
   1& 0
   \end{bmatrix} \cap \left[\mathcal{C}_2^\perp \ \mathcal{C}_4^\perp \right] \begin{bmatrix}
   1& 1\\
   1& 0
   \end{bmatrix}.
\end{equation*}
We notice that $\mathcal{C}^\perp$ also has the best-known parameters $[10,3,5]$.
\end{example}

\begin{example}
Let $\theta\in\mathbb{F}_4$ be such that $\theta^2+\theta+1=0$ and let
\begin{align*}
\mathbf{G}_1&=\begin{bmatrix}
   1 & 1 & 0 & 0 & 0&\theta^2&\theta^2&\theta
\end{bmatrix},
&\mathbf{G}_2=\begin{bmatrix}
   1 & 0 & 0&\theta^2&\theta & 0 & 1&\theta\\
   0 & 1 & 0&\theta&\theta&\theta^2&\theta^2&\theta\\
   0 & 0 & 1 & 1 & 1&\theta^2&\theta&\theta
\end{bmatrix} , \\
\mathbf{G}_3&=\begin{bmatrix}
   1 & 0 & 0&\theta^2 & 0&\theta^2&\theta&\theta^2\\
   0 & 1 & 0&\theta & 0 & 0 & 1&\theta^2\\
   0 & 0 & 1 & 0 & 0 & 0 & 0&\theta\\
   0 & 0 & 0 & 0 & 1 & 1 & 0&\theta^2
\end{bmatrix},
&\mathbf{G}_4=\begin{bmatrix}
   1 & 0 & 0 & 0 & 0 & 1 & 0 & 1\\
   0 & 1 & 0 & 0 & 0 & 1 & 0 & 1\\
   0 & 0 & 1 & 0 & 0& \theta & 0&\theta\\
   0 & 0 & 0 & 1 & 0&\theta^2 & 0 & 1\\
   0 & 0 & 0 & 0 & 1& \theta^2 & 0 & 1\\
   0 & 0 & 0 & 0 &0 & 0 & 1&\theta^2\\
\end{bmatrix}  , \\
\mathbf{G}_5&=\begin{bmatrix}
   1 & 0 & 0 & 0 & 0 & 0 & 1&\theta^2\\
   0 & 1 & 0 & 0 & 0 & 0 & 1 & 0\\
   0 & 0 & 1 & 0 & 0 & 0&\theta&\theta\\
  0 & 0 & 0 & 1 & 0 & 0 & 0&\theta\\
   0 & 0 & 0 & 0 & 1 & 0&  1&  0\\
   0 & 0 & 0 & 0 & 0 & 1 & 0&\theta^2
\end{bmatrix},
&\mathcal{A}=\begin{bmatrix}
   1 &  0 &\theta\\
   0  & 1 & 0\\
 \theta^2 & 0 &\theta\\
   0 &\theta^2 &\theta^2\\
   1 &  1  & 1
   \end{bmatrix}.
\end{align*}
Consider the MP code 
\begin{equation*}
\mathcal{C}=\left[\mathcal{C}_1 \ \mathcal{C}_2 \ \mathcal{C}_3 \ \mathcal{C}_4 \ \mathcal{C}_5 \right] \mathcal{A},
\end{equation*}
where $\mathcal{C}_1$, $\mathcal{C}_2$, $\mathcal{C}_3$, $\mathcal{C}_4$, and $\mathcal{C}_5$ are the linear codes over $\mathbb{F}_4$ of length $8$ generated by $\mathbf{G}_1$, $\mathbf{G}_2$, $\mathbf{G}_3$, $\mathbf{G}_4$, and $\mathbf{G}_5$, respectively. According to \cite{Grassl:codetables}, $\mathcal{C}$ has the best-known parameters $[24,20,3]$. By Theorems \ref{dualMP} and \ref{dualMPCase2},
\begin{equation*}
\begin{split}
\mathcal{C}&=\left[\mathcal{C}_1 \ \mathcal{C}_2 \ \mathcal{C}_3 \right] \begin{bmatrix}
   1 &  0 &\theta\\
   0  & 1 & 0\\
 \theta^2 & 0 &\theta
   \end{bmatrix} + \left[\mathcal{C}_4 \ \mathcal{C}_5 \right] \begin{bmatrix}
   0 &\theta^2 &\theta^2\\
   1 &  1  & 1
   \end{bmatrix}\\
   &=\left[\mathcal{C}_1 \ \mathcal{C}_2 \ \mathcal{C}_3 \right] \begin{bmatrix}
   1 &  0 &\theta\\
   0  & 1 & 0\\
 \theta^2 & 0 &\theta
   \end{bmatrix} + \left[\mathcal{C}_4 \ \mathcal{C}_5  \ \mathcal{O} \right] \begin{bmatrix}
   0 &\theta^2 &\theta^2\\
   1 &  1  & 1\\
   0 &  1  & 0
   \end{bmatrix}.
\end{split}
\end{equation*}
Therefore,
\begin{equation*}
\mathcal{C}^\perp=\left[\mathcal{C}_1^\perp \ \mathcal{C}_2^\perp \ \mathcal{C}_3^\perp \right] \begin{bmatrix}
 \theta^2 & 0 & 1\\
   0 & 1 & 0\\
 \theta^2 & 0&\theta
   \end{bmatrix} \cap \left[\mathcal{C}_4^\perp \ \mathcal{C}_5^\perp \ \mathcal{F} \right] \begin{bmatrix}
 \theta & 0&\theta\\
   1 & 0 & 0\\
   0 & 1 & 1
   \end{bmatrix}.
\end{equation*}
We notice that $\mathcal{C}^\perp$ has the parameters $[24,4,15]$ over $\mathbb{F}_4$, hence, it is almost optimal according to \cite{Grassl:codetables}. 
\end{example}

\section{Self-orthogonal MP codes}
\label{selforthgof MP}
In this section, we apply the results of Section \ref{EandGdualMP} to present necessary and sufficient conditions for an MP code to be Galois self-orthogonal. Current studies in the literature on self-orthogonality of MP codes require specific forms for the defining matrices. For instance, in \cite{Jitman,JitmanConf}, $\mathcal{A}\mathcal{A}^{\dagger}$ must be diagonal or anti-diagonal, where $^{\dagger}$ denotes the Hermitian transpose. In \cite{Zhang2024}, $\mathcal{A}$ must satisfy the partitioned Hermitian orthogonal property, whereas \cite{Edgar} studied MP codes associated with Toeplitz square matrices. Our study on Galois self-orthogonality of MP codes is more general in the sense that we do not require special forms, and the defining matrix does not have to be of full row rank. The main objective of this section is to prove the following.

\begin{theorem}
\label{GaloisSelforthMP_Total}
Let $\mathcal{C}=\left[\mathcal{C}_1 \ \mathcal{C}_2 \ \cdots \ \mathcal{C}_M\right]\mathcal{A}$, where $\mathcal{A}\in \mathbb{F}_q^{M\times N}$, $q=p^e$ for some prime $p$, and $\mathcal{C}_1 , \mathcal{C}_2 , \ldots , \mathcal{C}_M$ are linear codes of length $n$ over $\mathbb{F}_q$. For $0\le \ell <e$, the MP code $\mathcal{C}$ is $\ell$-Galois self-orthogonal if and only if $\mathcal{C}_i \subseteq \mathcal{C}_j^{\perp_\ell}$ for every nonzero $(i,j)^{\text{th}}$ entry of $\sigma^\ell\left(\mathcal{A}\right)\mathcal{A}^T$. In other words, if $\sigma^\ell\left(\mathcal{A}\right)\mathcal{A}^T=\left[\zeta_{i,j}\right]$, then $\mathcal{C}$ is $\ell$-Galois self-orthogonal if and only if $\mathcal{C}_i \subseteq \mathcal{C}_j^{\perp_\ell}$ for every $\zeta_{i,j}\ne 0$.
\end{theorem}
To prove Theorem \ref{GaloisSelforthMP_Total}, we consider two cases separately: $\mathcal{A}$ is of full row rank, and $\mathcal{A}$ has a rank less than $M$. Each case is considered in a subsection.

\subsection{$\mathcal{A}$ of rank $M$}
\label{SO_RankM}
In this subsection, we prove Theorem \ref{GaloisSelforthMP_Total} for an MP code with a defining matrix $\mathcal{A}\in\mathbb{F}_q^{M\times N}$ of rank $M$. To achieve the desired result, we first establish a more general setting from which the theorem can be easily deduced.

\begin{lemma}
\label{lem_original}
Let $\mathcal{C}^{(1)}=\left[\mathcal{C}_1^{(1)} \ \mathcal{C}_2^{(1)} \ \cdots \ \mathcal{C}_{M_1}^{(1)}\right]\mathcal{A}_1$ and $\mathcal{C}^{(2)}=\left[\mathcal{C}_1^{(2)} \ \mathcal{C}_2^{(2)} \ \cdots \ \mathcal{C}_{M_2}^{(2)}\right]\mathcal{A}_2$, where $\mathcal{A}_1\in\mathbb{F}_q^{M_1\times N}$ and $\mathcal{A}_2\in\mathbb{F}_q^{M_2\times N}$ are of full row rank. Then $\mathcal{C}^{(1)} \subseteq \mathcal{C}^{(2)\perp_\ell}$ if and only if $\mathcal{C}_i^{(1)} \subseteq \mathcal{C}^{(2)\perp_\ell}_j$ for every nonzero $(i,j)^{\text{th}}$ entry of $\sigma^{\ell}\left(\mathcal{A}_1\right)\mathcal{A}_2^T$.
\end{lemma}
\begin{proof}
Choose an invertible matrix $\mathcal{B}_2 \in\mathbb{F}_q^{N\times N}$ such that $\mathcal{B}_2 \{1,\ldots,M_2\}=\mathcal{A}_2$. Let $\mathcal{C}_{M_2+1}^{(2)}=\cdots=\mathcal{C}_N^{(2)}=\mathcal{O}$, then Theorem \ref{dualMP} gives
\begin{align*}
 \mathcal{C}^{(2)\perp_\ell}&= \sigma^{e-\ell}\left(  \mathcal{C}^{(2)\perp}\right)= \sigma^{e-\ell}\left( \left[\mathcal{C}^{(2)\perp}_1 \ \mathcal{C}^{(2)\perp}_2 \ \cdots \ \mathcal{C}^{(2)\perp}_N \right]\left(\mathcal{B}_2^{-1}\right)^T\right).
\end{align*}
Observe that $\mathcal{C}^{(1)} \subseteq \mathcal{C}^{(2)\perp_\ell}$ if and only if $\sigma^{\ell}\left(\mathcal{C}^{(1)}\right)\mathcal{B}_2^T \subseteq \sigma^{\ell}\left(\mathcal{C}^{(2)\perp_\ell}\right)\mathcal{B}_2^T$. Equivalently, 
$$\left[\sigma^{\ell}\mathcal{C}_1^{(1)} \ \sigma^{\ell}\mathcal{C}_2^{(1)} \ \cdots \ \sigma^{\ell}\mathcal{C}_{M_1}^{(1)}\right]\sigma^{\ell}\left(\mathcal{A}_1\right)\mathcal{B}_2^T \subseteq \left[\mathcal{C}^{(2)\perp}_1 \ \mathcal{C}^{(2)\perp}_2 \ \cdots \ \mathcal{C}^{(2)\perp}_N \right].$$
This is true if and only if 
\begin{equation}
\label{Inproof1}
\left[\sigma^{\ell}\mathcal{C}_1^{(1)} \ \sigma^{\ell}\mathcal{C}_2^{(1)} \ \cdots \ \sigma^{\ell}\mathcal{C}_{M_1}^{(1)}\right]\sigma^{\ell}\left(\mathcal{A}_1\right)\text{ Col}_j\left(\mathcal{B}_2^T\right) \subseteq \mathcal{C}^{(2)\perp}_j,
\end{equation}
for each $1\le j\le N$, where $\text{Col}_j$ denotes the $j^{\text{th}}$ column. For $M_2< j\le N$, \eqref{Inproof1} is satisfied because $\mathcal{C}^{(2)\perp}_j=\mathcal{F}$, thus we may restrict ourselves to $1\le j\le M_2$ and replace $\text{Col}_j\left(\mathcal{B}_2^T\right)$ by $\text{Col}_j\left(\mathcal{A}_2^T\right)$. At this point, we have $\mathcal{C}^{(1)} \subseteq \mathcal{C}^{(2)\perp_\ell}$ if and only if
\begin{equation*}
\left[\sigma^{\ell}\mathcal{C}_1^{(1)} \ \sigma^{\ell}\mathcal{C}_2^{(1)} \ \cdots \ \sigma^{\ell}\mathcal{C}_{M_1}^{(1)}\right]\text{ Col}_j\left(\sigma^{\ell}\left(\mathcal{A}_1\right)\mathcal{A}_2^T\right) \subseteq \mathcal{C}^{(2)\perp}_j,
\end{equation*}
for each $1\le j\le M_2$. The linearity of $\mathcal{C}^{(2)\perp}_j$ implies that the last condition is true if and only if  $\sigma^{\ell}\mathcal{C}_i^{(1)} \subseteq \mathcal{C}^{(2)\perp}_j$ for every nonzero $(i,j)^{\text{th}}$ entry in $\sigma^{\ell}\left(\mathcal{A}_1\right)\mathcal{A}_2^T$; but $\sigma^{\ell}\mathcal{C}_i^{(1)} \subseteq \mathcal{C}^{(2)\perp}_j$ if and only if $\mathcal{C}_i^{(1)} \subseteq \mathcal{C}^{(2)\perp_\ell}_j$.
\end{proof}

The main result of this subsection may now be concluded as a special case of Lemma \ref{lem_original} by setting $\mathcal{C}^{(1)}=\mathcal{C}^{(2)}=\mathcal{C}$.

\begin{theorem}
\label{GaloisSelforthMP}
Let $\mathcal{C}=\left[\mathcal{C}_1 \ \mathcal{C}_2 \ \cdots \ \mathcal{C}_M\right]\mathcal{A}$, where $\mathcal{A}\in \mathbb{F}_q^{M\times N}$ has full row rank, $\mathcal{C}_1 , \mathcal{C}_2 , \ldots , \mathcal{C}_M$ are linear codes of length $n$ over $\mathbb{F}_q$, and $q=p^e$. For $0\le \ell <e$, $\mathcal{C}$ is $\ell$-Galois self-orthogonal if and only if $\mathcal{C}_i \subseteq \mathcal{C}_j^{\perp_\ell}$ for every nonzero $(i,j)^{\text{th}}$ entry of $\sigma^\ell\left(\mathcal{A}\right)\mathcal{A}^T$. 
\end{theorem}
\begin{proof}
If we set $\mathcal{C}^{(1)}=\mathcal{C}^{(2)}=\mathcal{C}$, then $M_1=M_2=M$, $\mathcal{C}_i^{(1)}=\mathcal{C}_i^{(2)}=\mathcal{C}_i$, and $\mathcal{A}_1=\mathcal{A}_2=\mathcal{A}$.  Lemma \ref{lem_original} shows that $\mathcal{C} \subseteq \mathcal{C}^{\perp_\ell}$ if and only if $\mathcal{C}_i \subseteq \mathcal{C}^{\perp_\ell}_j$ for every nonzero $(i,j)^{\text{th}}$ entry of $\sigma^{\ell}\left(\mathcal{A}\right)\mathcal{A}^T$.
\end{proof}

\begin{example}
\label{ExGSO}
Let $\theta\in\mathbb{F}_4$ be such that $\theta^2+\theta+1=0$ and let
\begin{align*}
\mathbf{G}_1=\begin{bmatrix}
   1 &  0 &  0 &\theta &\theta^2\\
   0  & 1  & 0  & 1& \theta^2\\
   0  & 0 &  1 &  1 &  1
\end{bmatrix}\quad , \quad
\mathbf{G}_2=\begin{bmatrix}
   1  & 0  & 1 &\theta &\theta^2\\
   0  & 1  & 1& \theta^2& \theta
\end{bmatrix}.
\end{align*}
Consider the MP code $\mathcal{C}=\left[\mathcal{C}_1 \ \mathcal{C}_2 \right]\mathcal{A}$, where $\mathcal{C}_1$, $\mathcal{C}_2$ are the linear codes of length $5$ over $\mathbb{F}_4$ generated by $\mathbf{G}_1$, $\mathbf{G}_2$, respectively, and $\mathcal{A}$ is the NSC matrix
\begin{equation*}
\mathcal{A}= \begin{bmatrix}
    \theta^2 &  1 &\theta^2 &  1\\
   0 &  1 &\theta &\theta
   \end{bmatrix}.
\end{equation*}
We claim that $\mathcal{C}$ is a $1$-Galois self-orthogonal. Indeed,
\begin{equation*}
\sigma\left(\mathcal{A}\right)\mathcal{A}^T = \begin{bmatrix}
   0 &  0\\
   0 &  1
   \end{bmatrix}.
\end{equation*}
By Theorem \ref{GaloisSelforthMP}, $\mathcal{C}$ is $1$-Galois self-orthogonal because $\mathcal{C}_2 \subseteq \sigma\left(\mathcal{C}_2^\perp \right)$, to see this, note that $\mathbf{G}_2 \left(\sigma\mathbf{G}_2\right)^T=\mathbf{0}$. Moreover, $\mathcal{C}$ has the best-known parameters $[20,5,12]$ over $\mathbb{F}_4$ according to \cite{Grassl:codetables}. 
\end{example}

The Euclidean self-orthogonality condition for an MP code with a defining matrix of full row rank may now be easily obtained from Theorem \ref{GaloisSelforthMP} by letting $\ell=0$.

\begin{theorem}
\label{SelforthMP}
Let $\mathcal{C}=\left[\mathcal{C}_1 \ \mathcal{C}_2 \ \cdots \ \mathcal{C}_M\right]\mathcal{A}$, where $\mathcal{C}_1 , \mathcal{C}_2 , \ldots , \mathcal{C}_M$ are linear codes over $\mathbb{F}_q$, and $\mathcal{A}\in \mathbb{F}_q^{M\times N}$ has full row rank. Then $\mathcal{C}$ is Euclidean self-orthogonal if and only if $\mathcal{C}_i \subseteq \mathcal{C}_j^\perp$ for every nonzero $(i,j)^{\text{th}}$ entry of $\mathcal{A}\mathcal{A}^T$. 
\end{theorem}

Theorem \ref{SelforthMP} generalizes \cite[Theorems III.1 and III.4]{JitmanConf} by relaxing the requirement of an invertible $\mathcal{A}$ with diagonal or anti-diagonal $\mathcal{A} \mathcal{A}^T$. This is illustrated in Example \ref{ex_variety}, where a Euclidean self-orthogonal MP code is constructed using a non-square $\mathcal{A}$ whose product $\mathcal{A}\mathcal{A}^T$ is neither diagonal nor anti-diagonal. Moreover, the constituent codes show varying properties\textemdash some are self-orthogonal, some are dual-containing, and others are neither self-orthogonal nor dual-containing.

\begin{example}
\label{ex_variety}
Let
\begin{equation*}
\mathcal{A}= \begin{bmatrix}
 1&1&0&0\\
 0&1&0&0\\
 0&0&1&1
   \end{bmatrix}.
\end{equation*}
Consider any binary MP code $\mathcal{C}=\left[\mathcal{C}_1 \ \mathcal{C}_2 \ \mathcal{C}_3 \right]\mathcal{A}$ with $\mathcal{C}_1=\mathcal{C}_2^\perp$ and a self-orthogonal $\mathcal{C}_2$. Theorem \ref{SelforthMP} implies that $\mathcal{C}$ is Euclidean self-orthogonal. This is because 
\begin{equation*}
\mathcal{A} \mathcal{A}^T= \begin{bmatrix}
 0&1&0\\
 1&1&0\\
 0&0&0
   \end{bmatrix}.
\end{equation*}
It should be noted that $\mathcal{C}$ is Euclidean self-orthogonal, regardless of the choice of $\mathcal{C}_3$. Also, $\mathcal{C}_1$ is dual-containing because $\mathcal{C}_2$ is self-orthogonal. 
\end{example}

Several special cases for the self-orthogonality of MP codes can be obtained as corollaries of Theorem \ref{GaloisSelforthMP_Total}. As appeared in the literature, particular forms of $\mathcal{A}\mathcal{A}^T$ are required in such cases. Some of these cases is mentioned below in the context of Galois dual.

\begin{corollary}
If $\mathcal{C}=\left[\mathcal{C}_1 \ \mathcal{C}_2 \ \cdots \ \mathcal{C}_M\right]\mathcal{A}$ is such that $\sigma^\ell\left(\mathcal{A}\right)\mathcal{A}^T=\mathbf{0}$, then $\mathcal{C}$ is $\ell$-Galois self-orthogonal.
\end{corollary}
\begin{proof}
This is clear from Theorem \ref{SelforthMP}, although an alternate proof is presented. From \eqref{MP_Gen}, $\mathcal{C}$ has the generator matrix $\mathbf{G}=\text{diag}\left[\mathbf{G}_1 \ \cdots \ \mathbf{G}_M\right] ( \mathcal{A}\otimes \mathcal{I}_n)$. If $\sigma^\ell\left(\mathcal{A}\right)\mathcal{A}^T=\mathbf{0}$, then $\sigma^\ell\left(\mathbf{G}\right)\mathbf{G}^T=\mathbf{0}$, and hence, $\mathcal{C}$ is $\ell$-Galois self-orthogonal.
\end{proof}

\begin{corollary}
Let $\mathcal{C}=\left[\mathcal{C}_1 \ \mathcal{C}_2 \ \cdots \ \mathcal{C}_M\right]\mathcal{A}$ with $\mathcal{A}\in \mathbb{F}_q^{M\times N}$ such that $\sigma^\ell\left(\mathcal{A}\right)\mathcal{A}^T$ is diagonal. Then $\mathcal{C}$ is $\ell$-Galois self-orthogonal if and only if $\mathcal{C}_i$ is $\ell$-Galois self-orthogonal for every nonzero $i^{\text{th}}$ diagonal entry of $\sigma^\ell\left(\mathcal{A}\right)\mathcal{A}^T$.
\end{corollary}

\begin{corollary}
\label{Corr_anti_diag}
Let $\mathcal{C}=\left[\mathcal{C}_1 \ \mathcal{C}_2 \ \cdots \ \mathcal{C}_M\right]\mathcal{A}$ with $\mathcal{A}\in \mathbb{F}_q^{M\times N}$ such that $\sigma^\ell\left(\mathcal{A}\right)\mathcal{A}^T$ is anti-diagonal. Then $\mathcal{C}$ is $\ell$-Galois self-orthogonal if and only if $\mathcal{C}_i \subseteq \mathcal{C}_{M-i+1}^{\perp_\ell}$ for every nonzero $i^{\text{th}}$ anti-diagonal entry of $\sigma^\ell\left(\mathcal{A}\right)\mathcal{A}^T$. 
\end{corollary}

We conclude this subsection by showing the existence of Euclidean self-orthogonal MP codes with the best-known parameters. The following example offers a best-known parameters binary MP code in the sense that a binary self-orthogonal code must have an even minimum distance.

\begin{example}
Consider a binary MP code $\mathcal{C}=\left[\mathcal{C}_1 \ \mathcal{C}_2 \right]\mathcal{A}$ with 
\begin{equation*}
\mathcal{A}= \begin{bmatrix}
 1 &1&1&0&1\\
 1&1&0&1&1
  \end{bmatrix}.
\end{equation*}
By Corollary \ref{Corr_anti_diag}, $\mathcal{C}$ is self-orthogonal if and only if $\mathcal{C}_1 \subseteq \mathcal{C}_{2}^\perp$ since $\mathcal{A} \mathcal{A}^T$ is the backward identity matrix. For instance, if $\mathcal{C}_1$ is the repetition code of length $9$ and $\mathcal{C}_2$ is generated by
\begin{equation*}
\mathbf{G}_2= \begin{bmatrix}
 1&0&1&1&0&1&0&1&1\\
 0&1&0&1&1&1&1&0&1
   \end{bmatrix},
\end{equation*}
then $\mathcal{C}$ has the best-known parameters $[45,3,24]$ as a self-orthogonal binary linear code.
\end{example}

\subsection{$\mathcal{A}$ of rank $<M$}
\label{SOwithsmallrank}
This subsection proves Theorem \ref{GaloisSelforthMP_Total} for an MP code with a non-full row rank defining matrix. As discussed after Theorem \ref{dualMPCase2}, partitioning the set of rows of $\mathcal{A}$ to independent sets leads to writing $\mathcal{C}$ as a sum $\sum_{t} \mathcal{C}^{(t)}$, where each $\mathcal{C}^{(t)}$ has a full row rank defining matrix. We begin by examining the condition for the self-orthogonality of $\mathcal{C}$ as conditions on the codes $\mathcal{C}^{(t)}$.

\begin{lemma}
\label{lem1}
Let $\mathcal{C}=\left[\mathcal{C}_1 \ \mathcal{C}_2 \ \cdots \ \mathcal{C}_M\right]\mathcal{A}$, where $\mathcal{C}_1, \mathcal{C}_2, \ldots, \mathcal{C}_M$ are linear codes over $\mathbb{F}_q$. If $\mathcal{C}=\sum_{t} \mathcal{C}^{(t)}$, then $\mathcal{C}$ is $\ell$-Galois self-orthogonal if and only if $\mathcal{C}^{(\iota)}\subseteq \mathcal{C}^{(\gamma) \perp_\ell}$ for every pair $(\iota,\gamma)$ (not necessarily distinct).
\end{lemma}
\begin{proof}
We note from Theorem \ref{dualMPCase2} that $\mathcal{C}$ is self-orthogonal if and only if
\begin{equation*}
\sum_{t} \mathcal{C}^{(t)}=\mathcal{C}\subseteq \mathcal{C}^{\perp_\ell} = \cap_{t} \mathcal{C}^{(t)\perp_\ell}.
\end{equation*}
The linearity of $\mathcal{C}^{(t)\perp_\ell}$ indicates that the latter condition holds if and only if $\mathcal{C}^{(\iota)} \subseteq \mathcal{C}^{(\gamma)\perp_\ell}$ for all pairs $(\iota,\gamma)$.
\end{proof}

We next investigate a necessary and sufficient condition for $\mathcal{C}^{(\iota)} \subseteq \mathcal{C}^{(\gamma)\perp_\ell}$ for some pair $(\iota,\gamma)$ as a consequence of Lemma \ref{lem_original}.

\begin{lemma}
\label{lem2}
Let $\mathcal{C}^{(\iota)}=\left[\mathcal{C}_1^{(\iota)} \ \mathcal{C}_2^{(\iota)} \ \cdots \ \mathcal{C}_{M_\iota}^{(\iota)}\right]\mathcal{A}_\iota$ and $\mathcal{C}^{(\gamma)}=\left[\mathcal{C}_1^{(\gamma)} \ \mathcal{C}_2^{(\gamma)} \ \cdots \ \mathcal{C}_{M_\gamma}^{(\gamma)}\right]\mathcal{A}_\gamma$. Then $\mathcal{C}^{(\iota)} \subseteq \mathcal{C}^{(\gamma)\perp_\ell}$ if and only if $\mathcal{C}_i^{(\iota)} \subseteq \mathcal{C}^{(\gamma)\perp_\ell}_j$ for every nonzero $(i,j)^{\text{th}}$ entry of $\sigma^\ell\left(\mathcal{A}_\iota\right)\mathcal{A}_\gamma^T$.
\end{lemma}
\begin{proof}
Set $\mathcal{C}^{(1)}=\mathcal{C}^{(\iota)}$, $\mathcal{C}^{(2)}=\mathcal{C}^{(\gamma)}$, $\mathcal{A}_1=\mathcal{A}_\iota$, and $\mathcal{A}_2=\mathcal{A}_\gamma$ in Lemma \ref{lem_original}.
\end{proof}

We are now prepared to prove Theorem \ref{GaloisSelforthMP_Total} for MP codes with any defining matrix.

\begin{proof}[Proof of Theorem \ref{GaloisSelforthMP_Total}]
By Theorem \ref{dualMPCase2}, $\mathcal{C}=\sum_t \mathcal{C}^{(t)}$, where $\mathcal{C}^{(t)}=\left[\mathcal{C}_1^{(t)} \ \mathcal{C}_2^{(t)} \ \cdots \ \mathcal{C}_{M_{t}}^{(t)}\right]\mathcal{A}_{t}$ and $\mathcal{A}_{t}$ is of full row rank. Without loss of generality, we may assume that $\mathcal{A}$ is a one-column block matrix with block entries $\mathcal{A}_{t}$; this can be done by permuting the constituent codes simultaneously with the rows of $\mathcal{A}$. By Lemmas \ref{lem1} and \ref{lem2}, $\mathcal{C}$ is $\ell$-Galois self-orthogonal if and only if $\mathcal{C}^{(\iota)}\subseteq \mathcal{C}^{(\gamma) \perp_\ell}$ for all pairs $(\iota,\gamma)$, that is, if and only if $\mathcal{C}_i^{(\iota)} \subseteq \mathcal{C}_j^{(\gamma)\perp_\ell}$ for every nonzero $(i,j)^{\text{th}}$ entry of $\sigma^\ell\left(\mathcal{A}_\iota\right) \mathcal{A}_\gamma^T$ for all pairs $(\iota,\gamma)$. But the $(\iota,\gamma)^{\text{th}}$ block entry of $\sigma^\ell\left(\mathcal{A}\right)\mathcal{A}^T$ is $\sigma^\ell\left(\mathcal{A}_\iota\right)\mathcal{A}_\gamma^T$. This proves that $\mathcal{C}$ is $\ell$-Galois self-orthogonal if and only if $\mathcal{C}_i^{(\iota)} \subseteq \mathcal{C}_j^{(\gamma\perp_\ell)}$ for every nonzero $(i,j)^{\text{th}}$ entry in each $(\iota,\gamma)$ block entry of $\sigma^\ell\left(\mathcal{A}\right)\mathcal{A}^T$. Equivalently, $\mathcal{C}_i \subseteq \mathcal{C}_j^{\perp_\ell}$ for every nonzero $(i,j)^{\text{th}}$ entry of $ \sigma^\ell\left(\mathcal{A}\right)\mathcal{A}^T$.
\end{proof}

\begin{example}
Let $\theta\in\mathbb{F}_4$ be such that $\theta^2+\theta+1=0$ and consider an MP code $\mathcal{C}=\left[\mathcal{C}_1 \ \mathcal{C}_2  \ \mathcal{C}_3  \ \mathcal{C}_4  \ \mathcal{C}_5 \right]\mathcal{A}$ with 
\begin{equation*}
\mathcal{A}= \begin{bmatrix}
   1  &1&\theta^2\\
 \theta&\theta&  0\\
   1 & 1 & 1\\
 \theta & 1&  0\\
 \theta&\theta^2&\theta^2
   \end{bmatrix}.
\end{equation*}
We aim to select constituent codes such that $\mathcal{C}$ is $1$-Galois self-orthogonal. For this purpose, we have
\begin{equation*}
\sigma\left(\mathcal{A}\right)\mathcal{A}^T= \begin{bmatrix}
   1  &0&\theta&\theta^2&  0\\
   0&  0&  0&\theta&\theta^2\\
 \theta^2&  0&  1&\theta^2&\theta\\
 \theta&\theta^2&\theta & 0&\theta\\
   0&\theta&\theta^2&\theta^2&  1
   \end{bmatrix}.
\end{equation*}
Since $e=2$, $\sigma=\sigma^{-1}$, and hence, $\mathcal{C}_i \subseteq \mathcal{C}_j^{\perp_1}$ if and only if $\mathcal{C}_j \subseteq \mathcal{C}_i^{\perp_1}$. By Theorem \ref{GaloisSelforthMP_Total}, $\mathcal{C}$ is $1$-Galois self-orthogonal if and only if $\mathcal{C}_1 \subseteq \mathcal{C}_1^{\perp_1}$, $\mathcal{C}_1 \subseteq \mathcal{C}_3^{\perp_1}$, $\mathcal{C}_1 \subseteq \mathcal{C}_4^{\perp_1}$, $\mathcal{C}_2 \subseteq \mathcal{C}_4^{\perp_1}$, $\mathcal{C}_2 \subseteq \mathcal{C}_5^{\perp_1}$, $\mathcal{C}_3 \subseteq \mathcal{C}_3^{\perp_1}$, $\mathcal{C}_3 \subseteq \mathcal{C}_4^{\perp_1}$, $\mathcal{C}_3 \subseteq \mathcal{C}_5^{\perp_1}$, $\mathcal{C}_4 \subseteq \mathcal{C}_5^{\perp_1}$, and $\mathcal{C}_5 \subseteq \mathcal{C}_5^{\perp_1}$. We choose the constituents that satisfy these conditions to be the linear codes generated by
\begin{align*}
\mathbf{G}_1&=\begin{bmatrix}
   1&\theta&\theta^2&  1&  0
\end{bmatrix},\\
\mathbf{G}_2&=\begin{bmatrix}
   1 & 0&\theta&\theta & 0
\end{bmatrix},\\
\mathbf{G}_3=\mathbf{G}_5&=\begin{bmatrix}
   1&  0 & 1&\theta^2&\theta
\end{bmatrix},\\
\mathbf{G}_4&=\begin{bmatrix}
   1&\theta^2& \theta&  0&  1
\end{bmatrix}.
\end{align*}
Then $\mathcal{C}$ is $1$-Galois self-orthogonal with the parameters $[15,5,4]$ over $\mathbb{F}_4$. However, $\mathcal{C}^{\perp_1}$ is $1$-Galois dual-containing with the parameters $[15,10,3]$, and hence almost optimal according to \cite{Grassl:codetables}. In the next section, we examine $\ell$-Galois dual-containing MP codes.
\end{example}

\section{Dual-containing MP codes}
\label{dualcontnof MP}
In this section, we examine the conditions in which an MP code is Galois dual-containing. Using the concept that a linear code is dual-containing if and only if its dual is self-orthogonal, we apply the results in Section \ref{selforthgof MP} to derive the dual-containment conditions. Our results are more general than \cite{Cao2020,Cao2024} because we do not require an invertible $\mathcal{A}$, nor do we assume that $\sigma^\ell\left(\mathcal{A}\right)\mathcal{A}^T$ is diagonal or monomial. Instead, we separately treat two distinct cases: one where $\mathcal{A}$ has full row rank and another where it does not. This provides a more comprehensive understanding of the Galois dual-containing properties of MP codes.

\subsection{$\mathcal{A}$ of rank $M$}
Let $q=p^e$ for a prime $p$ and a positive integer $e$. Throughout this subsection, we assume that $\mathcal{C}$ is an MP code with a full row rank defining matrix $\mathcal{A}$. We begin with the necessary and sufficient conditions for $\mathcal{C}$ to be $\ell$-Galois dual-containing for any $0\le \ell <e$. We exploit the fact that $\mathcal{C}$ is $\ell$-Galois dual-containing if and only if $\mathcal{C}^{\perp_\ell}$ is $(e-\ell)$-Galois self-orthogonal. This is because $\mathcal{C}\supseteq \mathcal{C}^{\perp_\ell}$ is equivalent to $\left(\mathcal{C}^{\perp_\ell}\right)^{\perp_{e-\ell}}\supseteq \mathcal{C}^{\perp_\ell}$. Thus we build our result on the self-orthogonality condition presented in Theorem \ref{GaloisSelforthMP}. By setting $\ell=0$, the Euclidean dual-containment condition is easily achieved.

\begin{theorem}
\label{GalDualContMP}
Let $\mathcal{C}=\left[\mathcal{C}_1 \ \mathcal{C}_2 \ \cdots \ \mathcal{C}_M\right]\mathcal{A}$, where $\mathcal{A}\in \mathbb{F}_q^{M\times N}$ has full row rank, $\mathcal{C}_1 , \mathcal{C}_2 , \ldots , \mathcal{C}_M$ are linear codes of length $n$ over $\mathbb{F}_q$, and $q=p^e$. Construct an invertible $\mathcal{B}\in\mathbb{F}_q^{N\times N}$ such that $\mathcal{B}\{1,\ldots,M\}=\mathcal{A}$. For some $0\le \ell <e$, let $\left(\sigma^{\ell}\left(\mathcal{B}\right)\mathcal{B}^T\right)^{-1}=\left[\zeta_{i,j}\right]$. Then $\mathcal{C}$ is $\ell$-Galois dual-containing if and only if all the following conditions are satisfied:
\begin{enumerate}
\item $\zeta_{i,j}=0$ for $i\ge M+1$ and $j\ge M+1$,
\item if $\zeta_{i,j}\ne 0$ with $i\le M$ and $j\ge M+1$, then $\mathcal{C}_i=\mathcal{F}$,
\item if $\zeta_{i,j}\ne 0$ with $i\ge M+1$ and $j\le M$, then $\mathcal{C}_j=\mathcal{F}$,
\item if $\zeta_{i,j}\ne 0$ with $i\le M$ and $j\le M$, then $\mathcal{C}_i^{\perp_\ell} \subseteq \mathcal{C}_j$.
\end{enumerate}
\end{theorem}
\begin{proof}
The property $\left(\mathcal{C}^{\perp_\ell}\right)^{\perp_{e-\ell}}=\mathcal{C}$ implies that $\mathcal{C}$ is $\ell$-Galois dual-containing if and only if $\mathcal{C}^{\perp_\ell}$ is $(e-\ell)$-Galois self-orthogonal. By Corollary \ref{CoroGaloisDual}, 
\begin{equation*}
\mathcal{C}^{\perp_\ell}=[\mathcal{C}^{\perp_\ell}_1 \ \mathcal{C}^{\perp_\ell}_2 \ \cdots \ \mathcal{C}^{\perp_\ell}_N]\left(\sigma^{e-\ell}\left(\mathcal{B}\right)^{-1}\right)^T, 
\end{equation*}
where $\mathcal{C}_{M+1}=\cdots=\mathcal{C}_N=\mathcal{O}$. Theorem \ref{GaloisSelforthMP} shows that $\mathcal{C}^{\perp_\ell}$ is $(e-\ell)$-Galois self-orthogonal if and only if 
$$\mathcal{C}_i^{\perp_\ell} \subseteq \left(\mathcal{C}_j^{\perp_\ell}\right)^{\perp_{e-\ell}}=\mathcal{C}_j$$ 
for every nonzero $(i,j)^{\text{th}}$ entry of 
\begin{align*}
\sigma^{e-\ell}\left(\sigma^{e-\ell}\left(\mathcal{B}\right)^{-1}\right)^T  \sigma^{e-\ell}\left(\mathcal{B}\right)^{-1} &= \sigma^{2(e-\ell)}\left(\left(\mathcal{B}^T\right)^{-1}  \sigma^\ell\left(\mathcal{B}\right)^{-1} \right)\\
&=\sigma^{2(e-\ell)} \left(\sigma^\ell\left(\mathcal{B}\right)\mathcal{B}^T\right)^{-1}.
\end{align*}
Since $\sigma$ is an automorphism of $\mathbb{F}_q$, the latter matrix and the matrix $\left(\sigma^\ell\left(\mathcal{B}\right)\mathcal{B}^T\right)^{-1}$ agree in the $(i,j)^{\text{th}}$ positions having nonzero entries. Therefore, $\mathcal{C}$ is $\ell$-Galois dual-containing if and only if $\mathcal{C}_i^{\perp_\ell} \subseteq \mathcal{C}_j$ for every $\zeta_{i,j}\ne 0$. Recalling that $\mathcal{C}_{M+1}=\cdots=\mathcal{C}_N=\mathcal{O}$, then
\begin{enumerate}
\item for $i\ge M+1$ and $j\ge M+1$, the impossibility of $\mathcal{F}= \mathcal{C}_i^{\perp_\ell} \subseteq \mathcal{C}_j=\mathcal{O}$ implies that $\zeta_{i,j}= 0$.
\item if $\zeta_{i,j}\ne 0$ with $i\le M$ and $j\ge M+1$, then $\mathcal{C}_i^{\perp_\ell} \subseteq \mathcal{C}_j=\mathcal{O}$ requires that $\mathcal{C}_i=\mathcal{F}$,
\item if $\zeta_{i,j}\ne 0$ with $i\ge M+1$ and $j\le M$, then $\mathcal{F}=\mathcal{C}_i^{\perp_\ell} \subseteq \mathcal{C}_j$ requires that $\mathcal{C}_j=\mathcal{F}$.
\end{enumerate}
\end{proof}

\begin{example}
Let $\theta\in\mathbb{F}_9$ be such that $\theta^2+2\theta+2=0$ and consider any MP code $\mathcal{C}=\left[\mathcal{C}_1 \ \mathcal{C}_2 \right]\mathcal{A}$, where $\mathcal{A}$ is the NSC matrix
\begin{equation*}
\mathcal{A}= \begin{bmatrix}
 \theta^7 &\theta &\theta^7\\
   2 &  1 &\theta^7
   \end{bmatrix}.
\end{equation*}
One possible choice for an invertible $\mathcal{B}\in\mathbb{F}_9^{3\times 3}$ such that $\mathcal{B}\{1,2\}=\mathcal{A}$ is
\begin{equation*}
\mathcal{B}= \begin{bmatrix}
 \theta^7 &\theta &\theta^7\\
   2 &  1 &\theta^7\\
1&0&0
   \end{bmatrix}.
\end{equation*}
For $\ell=1$, we have
\begin{equation*}
\left(\sigma^{\ell}\left(\mathcal{B}\right)\mathcal{B}^T\right)^{-1}= \begin{bmatrix}
   0& \theta& \theta\\
 \theta^3  & 1 &  0\\
 \theta^3 &  0 &  0
   \end{bmatrix}.
\end{equation*}
Theorem \ref{GalDualContMP} asserts that $\mathcal{C}$ is $1$-Galois dual-containing if and only if $\mathcal{C}_1=\mathcal{F}$ and $\mathcal{C}_2$ is $1$-Galois dual-containing. Consequently, $\mathcal{C}$ has the parameters $[3n, n+k_2, d]$, where $n$ is the length of the constituent codes, $k_2$ is the dimension of $\mathcal{C}_2$, and $d$ is the minimum distance, which is undoubtedly depending on the choice of the constituents.
\end{example}

Theorem \ref{GalDualContMP} has an important significance when the defining matrix is invertible. This assumption reduces the conditions of Theorem \ref{GalDualContMP} while generalizing the results in \cite{Cao2024} by eliminating the need for $\sigma^{\ell}\left(\mathcal{A}\right)\mathcal{A}^T$ to be diagonal or monomial. This is expressed in the following corollary, which is proved by setting $N=M$ in Theorem \ref{GalDualContMP}.

\begin{corollary}
\label{CorrDCGalois}
Let $\mathcal{C}=\left[\mathcal{C}_1 \ \mathcal{C}_2 \ \cdots \ \mathcal{C}_M\right]\mathcal{A}$, where $\mathcal{A}\in \mathbb{F}_q^{M\times M}$ is invertible, $\mathcal{C}_1 , \mathcal{C}_2 , \ldots , \mathcal{C}_M$ are linear codes of length $n$ over $\mathbb{F}_q$, and $q=p^e$. For $0\le \ell <e$, $\mathcal{C}$ is $\ell$-Galois dual-containing if and only if $\mathcal{C}_i^{\perp_\ell} \subseteq \mathcal{C}_j$ for every nonzero $(i,j)^{\text{th}}$ entry of $\left(\sigma^{\ell}\left(\mathcal{A}\right)\mathcal{A}^T\right)^{-1}$. 
\end{corollary}

\begin{example}
Let $\theta\in\mathbb{F}_9$ be such that $\theta^2+2\theta+2=0$ and consider any MP code $\mathcal{C}=\left[\mathcal{C}_1 \ \mathcal{C}_2 \ \mathcal{C}_3 \ \mathcal{C}_4 \right]\mathcal{A}$, where $\mathcal{A}$ is the invertible matrix
\begin{equation*}
\mathcal{A}= \begin{bmatrix}
   0&\theta&  0&  0\\
   1&  1&\theta^3&\theta^2\\
 \theta&  2&\theta^7&  0\\
   1&  0&\theta^2&\theta^2
   \end{bmatrix}.
\end{equation*}
For $\ell=1$, we have
\begin{equation*}
\left(\sigma^{\ell}\left(\mathcal{A}\right)\mathcal{A}^T\right)^{-1}= \begin{bmatrix}
   1&  0&\theta^7&\theta\\
   0&  0&  0&\theta\\
 \theta^5 & 0&  1&\theta^3\\
 \theta^3&\theta^3&\theta & 0
   \end{bmatrix}.
\end{equation*}
Corollary \ref{CorrDCGalois} asserts that $\mathcal{C}$ is $1$-Galois dual-containing if and only if $\mathcal{C}_1^{\perp_1} \subseteq \mathcal{C}_1$, $\mathcal{C}_1^{\perp_1} \subseteq \mathcal{C}_3$, $\mathcal{C}_1^{\perp_1} \subseteq \mathcal{C}_4$, $\mathcal{C}_2^{\perp_1} \subseteq \mathcal{C}_4$, $\mathcal{C}_3^{\perp_1} \subseteq \mathcal{C}_1$, $\mathcal{C}_3^{\perp_1} \subseteq \mathcal{C}_3$, $\mathcal{C}_3^{\perp_1} \subseteq \mathcal{C}_4$, $\mathcal{C}_4^{\perp_1} \subseteq \mathcal{C}_1$, $\mathcal{C}_4^{\perp_1} \subseteq \mathcal{C}_2$, and $\mathcal{C}_4^{\perp_1} \subseteq \mathcal{C}_3$. These conditions can be reduced to $\mathcal{C}_1^{\perp_1} \subseteq \mathcal{C}_1$ if we choose $\mathcal{C}_3=\mathcal{C}_4=\mathcal{F}$. Furthermore, if we choose $\mathcal{C}_1$ to be the $1$-Galois dual-containing code generated by $\mathbf{G}_1$ and choose $\mathcal{C}_2$ to be the code generated by $\mathbf{G}_2$, where
\begin{equation*}
\mathbf{G}_1= \begin{bmatrix}
   1&  0&  0&\theta^3&  1\\
   0&  1 & 0&\theta^6&\theta\\
   0&  0&  1&\theta^7&\theta^7
   \end{bmatrix} \quad \text{and}\quad \mathbf{G}_2= \begin{bmatrix}
   1&  0&  0&  0&\theta^3\\
   0&  1&  0&  0&\theta^2\\
   0&  0&  1&  0&\theta^3\\
   0&  0&  0&  1&\theta^3
   \end{bmatrix},
\end{equation*}
then $\mathcal{C}$ will have the best-known parameters $[20,17,3]$ over $\mathbb{F}_9$ according to \cite{Grassl:codetables}. We notice that $\mathcal{C}$ is $1$-Galois dual-containing even though one of its constituents, namely $\mathcal{C}_2$, is not. This shows how our conditions for dual-containment are superior to \cite[Theorem 2]{Cao2020}, which requires all constituents to be dual-containing. 
\end{example}

Another important special case of Theorem \ref{GalDualContMP} is when $\ell=0$. It provides the necessary and sufficient conditions for an MP code to be Euclidean dual-containing. This can be formulated as follows.

\begin{corollary}
\label{dualContMP}
Let $\mathcal{C}=\left[\mathcal{C}_1 \ \mathcal{C}_2 \ \cdots \ \mathcal{C}_M\right]\mathcal{A}$, where $\mathcal{A}\in \mathbb{F}_q^{M\times N}$ has full row rank and $\mathcal{C}_1 , \mathcal{C}_2 , \ldots , \mathcal{C}_M$ are linear codes of length $n$ over $\mathbb{F}_q$. Construct an invertible $\mathcal{B}\in\mathbb{F}_q^{N\times N}$ such that $\mathcal{B}\{1,\ldots,M\}=\mathcal{A}$. Let $\left(\mathcal{B} \mathcal{B}^T\right)^{-1}=\left[\zeta_{i,j}\right]$. Then $\mathcal{C}$ is Euclidean dual-containing if and only if all the following conditions are satisfied:
\begin{enumerate}
\item $\zeta_{i,j}=0$ for $i\ge M+1$ and $j\ge M+1$,
\item if $\zeta_{i,j}\ne 0$ with $i\le M$ and $j\ge M+1$, then $\mathcal{C}_i=\mathcal{F}$,
\item if $\zeta_{i,j}\ne 0$ with $i\le M$ and $j\le M$, then $\mathcal{C}_i^{\perp} \subseteq \mathcal{C}_j$.
\end{enumerate}
\end{corollary}
\begin{proof}
By setting $\ell=0$ in Theorem \ref{GalDualContMP}, we notice that the second and third conditions of Theorem \ref{GalDualContMP} are the same due to the symmetry of $\left(\mathcal{B} \mathcal{B}^T\right)^{-1}$, and therefore one of them can be omitted.
\end{proof}

When designing a dual-containing MP code with a given defining matrix, the second condition of Corollary \ref{dualContMP} forces some constituents to be the whole space code, which makes meeting the third condition considerably easier. We demonstrate this with the following example.

\begin{example}
Consider any MP code $\mathcal{C}=\left[\mathcal{C}_1 \ \mathcal{C}_2 \ \mathcal{C}_3\right]\mathcal{A}$ over $\mathbb{F}_5$, where $\mathcal{A}$ is the NSC matrix
\begin{equation*}
\mathcal{A}=\begin{bmatrix}
 1&1&2&2\\
 0&4&1&4\\
 1&4&2&3
\end{bmatrix}.
\end{equation*}
To design Euclidean dual-containing $\mathcal{C}$ using Corollary \ref{dualContMP}, we select an invertible $\mathcal{B}$ and compute $\left(\mathcal{B} \mathcal{B}^T\right)^{-1}$ as follows:
\begin{equation*}
\mathcal{B}=\begin{bmatrix}
 1&1&2&2\\
 0&4&1&4\\
 1&4&2&3\\
 1&0&0&0
\end{bmatrix} \quad \text{ and } \quad \left(\mathcal{B} \mathcal{B}^T\right)^{-1}=\begin{bmatrix}
 2&4&3&0\\
 4&0&1&0\\
 3&1&1&1\\
 0&0&1&0
\end{bmatrix}.
\end{equation*}
Then $\mathcal{C}$ is dual-containing if and only if $\mathcal{C}_3=\mathcal{F}$, $\mathcal{C}_1^{\perp} \subseteq \mathcal{C}_2$, and $\mathcal{C}_1$ be dual-containing. One possibility for the generators of constituents meeting these conditions are 
\begin{equation*}
\mathbf{G}_1=\begin{bmatrix}
 1&0&0&1&2\\
 0&1&0&3&2\\
 0&0&1&2&1
\end{bmatrix}\quad \text{and}
\quad
\mathbf{G}_2=\begin{bmatrix}
 1&0&0&1&3\\
 0&1&0&3&4\\
 0&0&1&2&0
\end{bmatrix}.
\end{equation*}
Such a choice offers a Euclidean dual-containing MP code with the parameters $[20,11,4]$ over $\mathbb{F}_5$.
\end{example}

We conclude this subsection with another special case, Euclidean dual-containing MP code with an invertible defining matrix. This arises as a special case of Corollary \ref{CorrDCGalois} when $\ell=0$, or Corollary \ref{dualContMP} when $N=M$, and therefore $\mathcal{B}=\mathcal{A}$. In this latter case, the three conditions of Corollary \ref{dualContMP} are reduced to only the last one. 

\begin{corollary}
\label{CorrDCa}
Let $\mathcal{C}=\left[\mathcal{C}_1 \ \mathcal{C}_2 \ \cdots \ \mathcal{C}_M\right]\mathcal{A}$, where $\mathcal{A}\in \mathbb{F}_q^{M\times M}$ is invertible and $\mathcal{C}_1 , \mathcal{C}_2 , \ldots , \mathcal{C}_M$ are linear codes of length $n$ over $\mathbb{F}_q$. Then $\mathcal{C}$ is Euclidean dual-containing if and only if $\mathcal{C}_i^{\perp} \subseteq \mathcal{C}_j$ for every nonzero $(i,j)^{\text{th}}$ entry of $\left(\mathcal{A} \mathcal{A}^T\right)^{-1}$.
\end{corollary}

Corollary \ref{CorrDCa} seems to be very restrictive, however it generalizes some results in the literature, such as \cite[Theorem 3.4]{Cao2024} and \cite[Theorem 1]{Cao2020}. To see this, assume $\mathcal{A}\mathcal{A}^T$ is monomial; that is, if $\mathcal{A}\mathcal{A}^T=\left[\eta_{i,j}\right]$, then, for each $1\le j\le M$, there exists $1\le i_j \le M$ such that $\eta_{i_j,j}\ne 0$ but $\eta_{i,j}=0$ for all $i\ne i_j$. Let $\left(\mathcal{A}\mathcal{A}^T\right)^{-1}=\left[\zeta_{i,j}\right]$. Since $\mathcal{A}\mathcal{A}^T$ is symmetric, we observe that $\left(\mathcal{A}\mathcal{A}^T\right)^{-1}$ is not only monomial, but also $\zeta_{i,j}\ne 0$ if and only if $\eta_{i,j}\ne 0$. Corollary \ref{CorrDCa} states that $\mathcal{C}$ is Euclidean dual-containing if and only if $\mathcal{C}_{i_j}^{\perp} \subseteq \mathcal{C}_j$ for $1\le j\le M$.

The following example utilizes Corollary \ref{CorrDCa} and a computer search to present a Euclidean dual-containing MP code with the best-known parameters.

\begin{example}
Let $\theta\in\mathbb{F}_8$ be such that $\theta^3+\theta+1=0$ and let
\begin{equation*}
\mathcal{A}=\begin{bmatrix}
 \theta^4&\theta&\theta^6&\theta^6 &\theta^2\\
   0&\theta^6&\theta^5&\theta^2&\theta^4\\
 \theta^5& 0&  0&\theta^2&\theta^6\\
 \theta^2&\theta^6&\theta^6&\theta^6&\theta\\
   1&\theta&\theta&\theta^2&  0
\end{bmatrix}, \quad
\mathbf{G}_3=\mathbf{G}_4=\begin{bmatrix}
   1 & 0 & 0 & 0&\theta^5\\
   0&  1 & 0 & 0&\theta^3\\
   0 & 0 & 1 & 0&\theta^4\\
   0 & 0 & 0 & 1&\theta^3
\end{bmatrix}, \quad
\mathbf{G}_5=\begin{bmatrix}
   1 & 0 & 0 & 0&\theta^2\\
   0  &1 & 0 & 0&\theta\\
   0 & 0 & 1 & 0&\theta^3\\
   0 & 0 & 0 & 1&\theta^6
\end{bmatrix}.
\end{equation*}
Consider the MP code $\mathcal{C}=\left[\mathcal{C}_1 \ \mathcal{C}_2 \ \mathcal{C}_3 \ \mathcal{C}_4 \ \mathcal{C}_5 \right]\mathcal{A}$ over $\mathbb{F}_8$, where $\mathcal{C}_1=\mathcal{C}_2=\mathcal{F}$ and $\mathcal{C}_i$ is the linear code over $\mathbb{F}_8$ generated by $\mathbf{G}_i$ for $i=3, 4, 5$. According to \cite{Grassl:codetables}, $\mathcal{C}$ has the best-known parameters $[25,22,3]$ over $\mathbb{F}_8$. Moreover, $\mathcal{C}$ is Euclidean dual-containing by Corollary \ref{CorrDCa}. This is because
\begin{equation*}
\left(\mathcal{A} \mathcal{A}^T\right)^{-1}=\begin{bmatrix}
 \theta&\theta^6&\theta^3&\theta^3&\theta^2\\
 \theta^6&\theta^2&  1&\theta^5&\theta^5\\
 \theta^3&  1&\theta^4&\theta&  0\\
 \theta^3&\theta^5&\theta&\theta^5&  0\\
 \theta^2&\theta^5&  0 & 0 & 0
\end{bmatrix}
\end{equation*}
and $\mathcal{C}_3$ is dual-containing. We remark that $\mathcal{C}$ is dual-containing, but one of its constituents, namely $\mathcal{C}_5$, is not.
\end{example}

\subsection{$\mathcal{A}$ of rank $<M$}
In this subsection, we aim to present a sufficient condition under which an MP code $\mathcal{C}$ with a non-full row rank defining matrix is Galois dual-containing. We utilize Theorem \ref{dualMPCase2} to write $\mathcal{C}$ as a sum of some MP codes $\mathcal{C}^{(t)}$ with full row rank defining matrices. The following lemma shows that it is sufficient for one $\mathcal{C}^{(\gamma)}$ to contain the Galois dual of another $\mathcal{C}^{(\iota)}$ to ensure that $\mathcal{C}$ is Galois dual-containing:

\begin{lemma}
\label{lem3}
Let $\mathcal{C}=\left[\mathcal{C}_1 \ \mathcal{C}_2 \ \cdots \ \mathcal{C}_M\right]\mathcal{A}$, where $\mathcal{C}_1 , \mathcal{C}_2 , \ldots , \mathcal{C}_M$ are linear codes of length $n$ over $\mathbb{F}_q$, $q=p^e$, and $\mathcal{A}\in \mathbb{F}_q^{M\times N}$. As in Theorem \ref{dualMPCase2}, write $\mathcal{C}$ as a sum $\sum_{t} \mathcal{C}^{(t)}$, where $\mathcal{C}^{(t)}$ is MP with a full row rank defining matrix $\mathcal{A}_t\in\mathbb{F}_q^{M_t \times N}$. If $\mathcal{C}^{(\iota) \perp_\ell}\subseteq \mathcal{C}^{(\gamma)}$ for some pair $(\iota,\gamma)$, then $\mathcal{C}$ is $\ell$-Galois dual-containing.
\end{lemma}
\begin{proof}
Assume that $\mathcal{C}^{(\iota) \perp_\ell}\subseteq \mathcal{C}^{(\gamma)}$ for some pair $(\iota,\gamma)$. Then 
\begin{equation*}
\mathcal{C}=\sum_t \mathcal{C}^{(t)}\supseteq \mathcal{C}^{(\gamma)}\supseteq \mathcal{C}^{(\iota) \perp_\ell} \supseteq  \cap_t \mathcal{C}^{(t)\perp_\ell}= \mathcal{C}^{\perp_\ell}.
\end{equation*}
\end{proof}

Lemma \ref{lem3} establishes a sufficient condition for $\mathcal{C}$ to be $\ell$-Galois dual-containing, and this condition can be transferred to the constituents by means of Lemma \ref{lem_original}. Precisely, we view $\mathcal{C}^{(\iota) \perp_\ell}\subseteq \mathcal{C}^{(\gamma)}$ as $\mathcal{C}^{(1)} \subseteq \mathcal{C}^{(2)\perp_\ell}$, where $\mathcal{C}^{(1)}=\mathcal{C}^{(\iota) \perp_\ell}$ and $\mathcal{C}^{(2)}=\mathcal{C}^{(\gamma)\perp_{e-\ell}}$, then we apply Lemma \ref{lem_original}.

\begin{lemma}
\label{lem4}
Let $\mathcal{C}^{(\iota)}=\left[\mathcal{C}_1^{(\iota)} \ \mathcal{C}_2^{(\iota)} \ \cdots \ \mathcal{C}_{M_\iota}^{(\iota)}\right]\mathcal{A}_\iota$ and $\mathcal{C}^{(\gamma)}=\left[\mathcal{C}_1^{(\gamma)} \ \mathcal{C}_2^{(\gamma)} \ \cdots \ \mathcal{C}_{M_\gamma}^{(\gamma)}\right]\mathcal{A}_\gamma$, where $\mathcal{A}_\iota\in\mathbb{F}_q^{M_\iota \times N}$ and $\mathcal{A}_\gamma\in\mathbb{F}_q^{M_\gamma \times N}$ have full row rank. Construct invertible matrices $\mathcal{B}_\iota, \mathcal{B}_\gamma\in\mathbb{F}_q^{N\times N}$ such that $\mathcal{B}_\iota \{1,\ldots,M_\iota\}=\mathcal{A}_\iota$ and $\mathcal{B}_\gamma \{1,\ldots,M_\gamma\}=\mathcal{A}_\gamma$. For some $0\le \ell <e$, let $\left(\sigma^{\ell}\left(\mathcal{B}_\gamma\right)\mathcal{B}_\iota^T\right)^{-1}=\left[\zeta_{i,j}\right]$. Then $\mathcal{C}^{(\iota) \perp_\ell}\subseteq \mathcal{C}^{(\gamma)}$ if and only if  all the following conditions are satisfied:
\begin{enumerate}
\item $\zeta_{i,j}=0$ for $i\ge M_\iota+1$ and $j\ge M_\gamma+1$,
\item if $\zeta_{i,j}\ne 0$ with $i\le M_\iota$ and $j\ge M_\gamma+1$, then $\mathcal{C}_i^{(\iota)}=\mathcal{F}$,
\item if $\zeta_{i,j}\ne 0$ with $i\ge M_\iota+1$ and $j\le M_\gamma$, then $\mathcal{C}_j^{(\gamma)}=\mathcal{F}$,
\item if $\zeta_{i,j}\ne 0$ with $i\le M_\iota$ and $j\le M_\gamma$, then $\mathcal{C}_i^{(\iota)\perp_\ell}  \subseteq \mathcal{C}_j^{(\gamma)}$.
\end{enumerate}
\end{lemma}
\begin{proof}
By Corollary \ref{CoroGaloisDual}, 
\begin{equation*}
\begin{split}
\mathcal{C}^{(\iota) \perp_\ell}&=[\mathcal{C}^{(\iota) \perp_\ell}_1 \ \mathcal{C}^{(\iota) \perp_\ell}_2 \ \cdots \ \mathcal{C}^{(\iota) \perp_\ell}_N ]\left(\sigma^{e-\ell}\left(\mathcal{B}_\iota\right)^{-1}\right)^T,\\
\mathcal{C}^{(\gamma) \perp_{e-\ell}}&=[\mathcal{C}^{(\gamma) \perp_{e-\ell}}_1 \ \mathcal{C}^{(\gamma) \perp_{e-\ell}}_2 \ \cdots \ \mathcal{C}^{(\gamma) \perp_{e-\ell}}_N ]\left(\sigma^{\ell}\left(\mathcal{B}_\gamma\right)^{-1}\right)^T,
\end{split}
\end{equation*}
where $\mathcal{C}_{M_\iota+1}^{(\iota)}=\cdots=\mathcal{C}_{N}^{(\iota)}=\mathcal{C}_{M_\gamma+1}^{(\gamma)}=\cdots=\mathcal{C}_{N}^{(\gamma)}=\mathcal{O}$. Lemma \ref{lem_original} with $\mathcal{C}^{(1)}=\mathcal{C}^{(\iota) \perp_\ell}$, $\mathcal{C}^{(2)}=\mathcal{C}^{(\gamma) \perp_{e-\ell}}$, $\mathcal{A}_1=\left(\sigma^{e-\ell}\left(\mathcal{B}_\iota\right)^{-1}\right)^T$, and $\mathcal{A}_2=\left(\sigma^{\ell}\left(\mathcal{B}_\gamma\right)^{-1}\right)^T$ asserts that $\mathcal{C}^{(\iota) \perp_\ell}\subseteq \mathcal{C}^{(\gamma)}$ if and only if $\mathcal{C}_i^{(\iota) \perp_\ell}\subseteq \mathcal{C}_j^{(\gamma)}$ for every nonzero $(i,j)^{\text{th}}$ entry of 
$$\sigma^{\ell}\left(\sigma^{e-\ell}\left(\mathcal{B}_\iota\right)^{-1}\right)^T  \sigma^\ell\left(\mathcal{B}_\gamma\right)^{-1} =\left(\sigma^{\ell}\left(\mathcal{B}_\gamma\right) \mathcal{B}_\iota^T \right)^{-1}.$$ 
Since $\mathcal{C}_i^{(\iota) \perp_\ell}=\mathcal{F}$ for $ M_\iota +1 \le i \le N$ and $\mathcal{C}_j^{(\gamma)}=\mathcal{O}$ for $M_\gamma +1  \le j \le N$, the four conditions follow in a similar manner as in the proof of Theorem \ref{GalDualContMP}.
\end{proof}

The main result of this subsection is obtained by combining Lemmas \ref{lem3} and \ref{lem4}. It offers a sufficient condition for any MP code to be dual-containing.

\begin{theorem}
\label{GDualContGen}
Let $\mathcal{C}=\left[\mathcal{C}_1 \ \mathcal{C}_2 \ \cdots \ \mathcal{C}_M\right]\mathcal{A}$, where $\mathcal{C}_1 , \mathcal{C}_2 , \ldots , \mathcal{C}_M$ are linear codes of length $n$ over $\mathbb{F}_q$, $q=p^e$, and $\mathcal{A}\in \mathbb{F}_q^{M\times N}$. As in Theorem \ref{dualMPCase2}, write $\mathcal{C}$ as a sum $\sum_{t} \mathcal{C}^{(t)}$, where $\mathcal{C}^{(t)}$ is MP with a full row rank defining matrix $\mathcal{A}_t\in\mathbb{F}_q^{M_t \times N}$. For a fixed $0\le \ell <e$, suppose there exist invertible matrices $\mathcal{B}_\iota$ and $\mathcal{B}_\gamma$ for some pair $\mathcal{C}^{(\iota)}$ and $\mathcal{C}^{(\gamma)}$ with defining matrices $\mathcal{A}_\iota=\mathcal{B}_\iota \{1,\ldots,M_\iota\}$ and $\mathcal{A}_\gamma=\mathcal{B}_\gamma \{1,\ldots,M_\gamma\}$, respectively, such that
\begin{enumerate}
\item $\zeta_{i,j}=0$ for $i\ge M_\iota+1$ and $j\ge M_\gamma+1$,
\item if $\zeta_{i,j}\ne 0$ with $i\le M_\iota$ and $j\ge M_\gamma+1$, then $\mathcal{C}_i^{(\iota)}=\mathcal{F}$,
\item if $\zeta_{i,j}\ne 0$ with $i\ge M_\iota+1$ and $j\le M_\gamma$, then $\mathcal{C}_j^{(\gamma)}=\mathcal{F}$,
\item if $\zeta_{i,j}\ne 0$ with $i\le M_\iota$ and $j\le M_\gamma$, then $\mathcal{C}_i^{(\iota)\perp_\ell}  \subseteq \mathcal{C}_j^{(\gamma)}$,
\end{enumerate}
where $\left(\sigma^{\ell}\left(\mathcal{B}_\gamma\right)\mathcal{B}_\iota^T\right)^{-1}= \left[\zeta_{i,j}\right]$. Then $\mathcal{C}$ is $\ell$-Galois dual-containing.
\end{theorem}
\begin{proof}
Lemma \ref{lem4} asserts that $\mathcal{C}^{(\iota) \perp_\ell}\subseteq \mathcal{C}^{(\gamma)}$ if the preceding conditions are met. In turn, Lemma \ref{lem3} shows that $\mathcal{C}$ is $\ell$-Galois dual-containing.
\end{proof}

As a remedy for the sophistication of Theorem \ref{GDualContGen}, we focus on a special case. In this special case, we assume that the conditions of Theorem \ref{GDualContGen} are satisfied for an equal pair $\mathcal{C}^{(\iota)}=\mathcal{C}^{(\gamma)}$ that has an invertible defining matrix $\mathcal{A}_\iota=\mathcal{A}_\gamma$. This case can be stated as follows:

\begin{corollary}
\label{last_corr}
Let $\mathcal{C}=\left[\mathcal{C}_1 \ \mathcal{C}_2 \ \cdots \ \mathcal{C}_M\right]\mathcal{A}$, where $\mathcal{C}_1 , \mathcal{C}_2 , \ldots , \mathcal{C}_M$ are linear codes of length $n$ over $\mathbb{F}_q$, $q=p^e$, and $\mathcal{A}\in \mathbb{F}_q^{M\times N}$ has rank $N$. For a fixed $0\le \ell <e$, suppose there exist distinct integers $1\le r_1, r_2, \ldots, r_N \le M$ such that $\mathcal{A}_\gamma=\mathcal{A}\{r_1, r_2, \ldots, r_N\}$ is invertible and $\mathcal{C}_{r_i}^{\perp_\ell}  \subseteq \mathcal{C}_{r_j}$ for every nonzero $(i,j)^{\text{th}}$ entry of $\left(\sigma^{\ell}\left(\mathcal{A}_\gamma\right)\mathcal{A}_\gamma^T\right)^{-1}$. Then $\mathcal{C}$ is $\ell$-Galois dual-containing.
\end{corollary}
\begin{proof}
As in Theorem \ref{dualMPCase2}, one can write $\mathcal{C}=\sum_{t} \mathcal{C}^{(t)}$ with $\mathcal{C}^{(\gamma)}=\left[\mathcal{C}_{r_1} \ \mathcal{C}_{r_2} \ \cdots \ \mathcal{C}_{r_N}\right]\mathcal{A}_\gamma$ appearing as a summand. By hypothesis and Corollary \ref{CorrDCGalois}, $\mathcal{C}^{(\gamma)}$ is $\ell$-Galois dual-containing. Then
\begin{equation*}
\mathcal{C}=\sum_t \mathcal{C}^{(t)}\supseteq \mathcal{C}^{(\gamma)}\supseteq \mathcal{C}^{(\gamma) \perp_\ell} \supseteq  \cap_t \mathcal{C}^{(t)\perp_\ell}= \mathcal{C}^{\perp_\ell}.
\end{equation*}
\end{proof}

The following example is a direct application of Corollary \ref{last_corr} when $\ell=0$. It demonstrates the existence of a Euclidean dual-containing MP code with a full column rank defining matrix that also has the best-known parameters.

\begin{example}
\label{Ex10}
Consider any MP code $\mathcal{C}=\left[\mathcal{C}_1 \ \mathcal{C}_2 \ \mathcal{C}_3 \ \mathcal{C}_4\right]\mathcal{A}$ over $\mathbb{F}_3$, where 
\begin{equation*}
\mathcal{A}=\begin{bmatrix}
 2 & 1 & 1\\
 2 & 0 & 2\\
 2 & 0 & 0\\
 0 & 0 & 1\\
\end{bmatrix}.
\end{equation*}
Let $\mathcal{A}_1=\mathcal{A}\{1,2,3\}$ and $\mathcal{A}_2=\mathcal{A}\{4\}$. Then, indeed, $\mathcal{C}=\mathcal{C}^{(1)}+\mathcal{C}^{(2)}$, where $\mathcal{C}^{(1)}=\left[\mathcal{C}_1 \ \mathcal{C}_2 \ \mathcal{C}_3\right]\mathcal{A}_1$ and $\mathcal{C}^{(2)}=\left[\mathcal{C}_4\right]\mathcal{A}_2$. Observe that
\begin{equation*}
\left(\mathcal{A}_1 \mathcal{A}_1^T\right)^{-1}=\begin{bmatrix}
 1&1&1\\
 1&2&0\\
 1&0&0
\end{bmatrix}.
\end{equation*}
By Corollary \ref{last_corr}, $\mathcal{C}$ is Euclidean dual-containing if $\mathcal{C}^{(1)}$ is Euclidean dual-containing. The latter follows by Corollary \ref{CorrDCGalois} if and only if $\mathcal{C}_2^{\perp} \subseteq \mathcal{C}_2$ and $\mathcal{C}_1^{\perp}$ is a subset of $\mathcal{C}_1$, $\mathcal{C}_2$, and $\mathcal{C}_3$. One possibility for the generators of constituents of length $6$ over $\mathbb{F}_3$ meeting these conditions are 
\begin{align*}
\mathbf{G}_1&=\begin{bmatrix}
 1&0&0&0&0&1\\
 0&1&0&0&0&2\\
 0&0&1&0&0&1\\
 0&0&0&1&0&2\\
 0&0&0&0&1&2
\end{bmatrix}, \quad
\mathbf{G}_2=\begin{bmatrix}
 1&0&0&0&0&2\\
 0&1&0&0&0&2\\
 0&0&1&0&0&2\\
 0&0& 0&1&0&2\\
 0&0&0&0&1&1       
\end{bmatrix},\\
\mathbf{G}_3&=\begin{bmatrix}
 1&2&1&2&2&2
\end{bmatrix}, \quad
\mathbf{G}_4=\begin{bmatrix}
 1&0&1&1&1&0
\end{bmatrix}.
\end{align*}
Such a choice offers a Euclidean dual-containing MP code with the best-known parameters $[18,12,4]$ over $\mathbb{F}_3$ according to \cite{Grassl:codetables}. 
\end{example}

\section{Conclusion}
\label{concl}
We provided formulas for the Galois dual of MP codes and investigated the necessary and sufficient conditions for such codes to be Galois self-orthogonal and Galois dual-containing. Our results outperform those in the literature because we examine MP codes in a more general form. Specifically, we make no restrictions on the constituent codes of the MP codes, other than that they are linear. Additionally, we did not assume a square invertible defining matrix; instead, we separately investigated the cases of a full row rank defining matrix and the case where the defining matrix has a smaller rank. Furthermore, we made no assumptions on the form of the product $\mathcal{A}\mathcal{A}^T$. As an application, we presented several numerical examples using MP codes with the best-known parameters.



\begin{thebibliography}{14}
\ifx \bisbn   \undefined \def \bisbn  #1{ISBN #1}\fi
\ifx \binits  \undefined \def \binits#1{#1}\fi
\ifx \bauthor  \undefined \def \bauthor#1{#1}\fi
\ifx \batitle  \undefined \def \batitle#1{#1}\fi
\ifx \bjtitle  \undefined \def \bjtitle#1{#1}\fi
\ifx \bvolume  \undefined \def \bvolume#1{\textbf{#1}}\fi
\ifx \byear  \undefined \def \byear#1{#1}\fi
\ifx \bissue  \undefined \def \bissue#1{#1}\fi
\ifx \bfpage  \undefined \def \bfpage#1{#1}\fi
\ifx \blpage  \undefined \def \blpage #1{#1}\fi
\ifx \burl  \undefined \def \burl#1{\textsf{#1}}\fi
\ifx \doiurl  \undefined \def \doiurl#1{\url{https://doi.org/#1}}\fi
\ifx \betal  \undefined \def \betal{\textit{et al.}}\fi
\ifx \binstitute  \undefined \def \binstitute#1{#1}\fi
\ifx \binstitutionaled  \undefined \def \binstitutionaled#1{#1}\fi
\ifx \bctitle  \undefined \def \bctitle#1{#1}\fi
\ifx \beditor  \undefined \def \beditor#1{#1}\fi
\ifx \bpublisher  \undefined \def \bpublisher#1{#1}\fi
\ifx \bbtitle  \undefined \def \bbtitle#1{#1}\fi
\ifx \bedition  \undefined \def \bedition#1{#1}\fi
\ifx \bseriesno  \undefined \def \bseriesno#1{#1}\fi
\ifx \blocation  \undefined \def \blocation#1{#1}\fi
\ifx \bsertitle  \undefined \def \bsertitle#1{#1}\fi
\ifx \bsnm \undefined \def \bsnm#1{#1}\fi
\ifx \bsuffix \undefined \def \bsuffix#1{#1}\fi
\ifx \bparticle \undefined \def \bparticle#1{#1}\fi
\ifx \barticle \undefined \def \barticle#1{#1}\fi
\bibcommenthead
\ifx \bconfdate \undefined \def \bconfdate #1{#1}\fi
\ifx \botherref \undefined \def \botherref #1{#1}\fi
\ifx \url \undefined \def \url#1{\textsf{#1}}\fi
\ifx \bchapter \undefined \def \bchapter#1{#1}\fi
\ifx \bbook \undefined \def \bbook#1{#1}\fi
\ifx \bcomment \undefined \def \bcomment#1{#1}\fi
\ifx \oauthor \undefined \def \oauthor#1{#1}\fi
\ifx \citeauthoryear \undefined \def \citeauthoryear#1{#1}\fi
\ifx \endbibitem  \undefined \def \endbibitem {}\fi
\ifx \bconflocation  \undefined \def \bconflocation#1{#1}\fi
\ifx \arxivurl  \undefined \def \arxivurl#1{\textsf{#1}}\fi
\csname PreBibitemsHook\endcsname

\bibitem[\protect\citeauthoryear{Blackmore and Norton}{2001}]{Blackmore2001}
\begin{barticle}
\bauthor{\bsnm{Blackmore}, \binits{T.}},
\bauthor{\bsnm{Norton}, \binits{G.}}:
\batitle{Matrix-product codes over $\mathbb{F}_q$}.
\bjtitle{Appl. Algebra Engrg. Comm. Comput.}
\bvolume{12}(\bissue{6}),
\bfpage{477}--\blpage{500}
(\byear{2001})
\doiurl{10.1007/pl00004226}
\end{barticle}
\endbibitem

\bibitem[\protect\citeauthoryear{Hernando et~al.}{2009}]{Hernando2009}
\begin{barticle}
\bauthor{\bsnm{Hernando}, \binits{F.}},
\bauthor{\bsnm{Lally}, \binits{K.}},
\bauthor{\bsnm{Ruano}, \binits{D.}}:
\batitle{Construction and decoding of matrix-product codes from nested codes}.
\bjtitle{Appl. Algebra Engrg. Comm. Comput.}
\bvolume{20}(\bissue{5–6}),
\bfpage{497}--\blpage{507}
(\byear{2009})
\doiurl{10.1007/s00200-009-0113-5}
\end{barticle}
\endbibitem

\bibitem[\protect\citeauthoryear{Hernando and Ruano}{2010}]{Hernando2010}
\begin{barticle}
\bauthor{\bsnm{Hernando}, \binits{F.}},
\bauthor{\bsnm{Ruano}, \binits{D.}}:
\batitle{New linear codes from matrix-product codes with polynomial units}.
\bjtitle{Adv. Math. Commun.}
\bvolume{4}(\bissue{3}),
\bfpage{363}--\blpage{367}
(\byear{2010})
\doiurl{10.3934/amc.2010.4.363}
\end{barticle}
\endbibitem

\bibitem[\protect\citeauthoryear{Eldin}{2024}]{Eldin2024}
\begin{botherref}
\oauthor{\bsnm{Taki~Eldin}, \binits{R.}}:
Matrix product and quasi-twisted codes in one class.
Computational and Applied Mathematics
\textbf{43}(2)
(2024)
\doiurl{10.1007/s40314-024-02612-x}
\end{botherref}
\endbibitem

\bibitem[\protect\citeauthoryear{Liu et~al.}{2017}]{Liu2017}
\begin{barticle}
\bauthor{\bsnm{Liu}, \binits{X.}},
\bauthor{\bsnm{Dinh}, \binits{H.Q.}},
\bauthor{\bsnm{Liu}, \binits{H.}},
\bauthor{\bsnm{Yu}, \binits{L.}}:
\batitle{On new quantum codes from matrix product codes}.
\bjtitle{Cryptogr. Commun.}
\bvolume{10}(\bissue{4}),
\bfpage{579}--\blpage{589}
(\byear{2017})
\doiurl{10.1007/s12095-017-0242-9}
\end{barticle}
\endbibitem

\bibitem[\protect\citeauthoryear{Cao and Cui}{2020a}]{Cao2020}
\begin{barticle}
\bauthor{\bsnm{Cao}, \binits{M.}},
\bauthor{\bsnm{Cui}, \binits{J.}}:
\batitle{New stabilizer codes from the construction of dual-containing
  matrix-product codes}.
\bjtitle{Finite Fields Appl.}
\bvolume{63},
\bfpage{101643}
(\byear{2020})
\doiurl{10.1016/j.ffa.2020.101643}
\end{barticle}
\endbibitem

\bibitem[\protect\citeauthoryear{Cao and Cui}{2020b}]{2Cao2020}
\begin{botherref}
\oauthor{\bsnm{Cao}, \binits{M.}},
\oauthor{\bsnm{Cui}, \binits{J.}}:
Construction of new quantum codes via {H}ermitian dual-containing matrix-product
  codes.
Quantum Inf. Process.
\textbf{19}(12)
(2020)
\doiurl{10.1007/s11128-020-02921-0}
\end{botherref}
\endbibitem

\bibitem[\protect\citeauthoryear{Jitman and Mankean}{2017}]{Jitman}
\begin{botherref}
\oauthor{\bsnm{Jitman}, \binits{S.}},
\oauthor{\bsnm{Mankean}, \binits{T.}}:
Matrix-product constructions for {H}ermitian self-orthogonal codes.
Chamchuri J. Math.
\textbf{9}(35--51)
(2017)
\end{botherref}
\endbibitem

\bibitem[\protect\citeauthoryear{Zhang}{2024}]{Zhang2024}
\begin{botherref}
\oauthor{\bsnm{Zhang}, \binits{X.}}:
{H}ermitian self-orthogonal matrix product codes and their applications to
  quantum codes.
Quantum Inf. Process.
\textbf{23}(3)
(2024)
\doiurl{10.1007/s11128-024-04314-z}
\end{botherref}
\endbibitem

\bibitem[\protect\citeauthoryear{Cao}{2024}]{Cao2024}
\begin{barticle}
\bauthor{\bsnm{Cao}, \binits{M.}}:
\batitle{On dual-containing, almost dual-containing matrix-product codes and
  related quantum codes}.
\bjtitle{Finite Fields Appl.}
\bvolume{96},
\bfpage{102400}
(\byear{2024})
\doiurl{10.1016/j.ffa.2024.102400}
\end{barticle}
\endbibitem

\bibitem[\protect\citeauthoryear{Mankean and Jitman}{2016}]{JitmanConf}
\begin{bchapter}
\bauthor{\bsnm{Mankean}, \binits{T.}},
\bauthor{\bsnm{Jitman}, \binits{S.}}:
\bctitle{Matrix-product constructions for self-orthogonal linear codes}.
In: \bbtitle{2016 12th International Conference on Mathematics, Statistics, and
  Their Applications (ICMSA)},
pp. \bfpage{6}--\blpage{10}.
\bpublisher{IEEE}, 
(\byear{2016}).
\doiurl{10.1109/ICMSA.2016.7954297}
\end{bchapter}
\endbibitem

\bibitem[\protect\citeauthoryear{Grassl}{2007}]{Grassl:codetables}
\begin{botherref}
\oauthor{\bsnm{Grassl}, \binits{M.}}:
{Bounds on the minimum distance of linear codes and quantum codes}.
Online available at \url{http://www.codetables.de}.
Accessed on 2024-07-30
(2007)
\end{botherref}
\endbibitem

\bibitem[\protect\citeauthoryear{Taki~Eldin}{2023}]{TakiEldin2023}
\begin{barticle}
\bauthor{\bsnm{Taki~Eldin}, \binits{R.}}:
\batitle{Two-sided {G}alois duals of multi-twisted codes}.
\bjtitle{J. Appl. Math. Comput.}
\bvolume{69}(\bissue{4}),
\bfpage{3459}--\blpage{3487}
(\byear{2023})
\doiurl{10.1007/s12190-023-01889-1}
\end{barticle}
\endbibitem

\bibitem[\protect\citeauthoryear{Li et~al.}{2024}]{Edgar}
\begin{botherref}
\oauthor{\bsnm{Li}, \binits{Y.}},
\oauthor{\bsnm{Zhu}, \binits{S.}},
\oauthor{\bsnm{Martínez-Moro}, \binits{E.}}:
On $\sigma$ self-orthogonal matrix-product codes associated with {T}oeplitz
  matrices.
arXiv
(2024).
\doiurl{10.48550/ARXIV.2405.06292} .
\url{https://arxiv.org/abs/2405.06292}
\end{botherref}
\endbibitem
\end{thebibliography}

\end{document}